\documentclass[a4paper,USenglish,thm-restate]{lipics-v2021}

\pdfoutput=1
\hideLIPIcs

\usepackage{complexity}
\usepackage[e]{esvect}
\DeclareFontFamily{U}{esvect}{}
\DeclareFontShape{U}{esvect}{m}{n}{
  <-5.5> vect5
  <5.5-6.5> vect6
  <6.5-7.5> vect7
  <7.5-8.5> vect8
  <8.5-9.5> vect9
  <9.5-> vect10
}{}
\usepackage{mathtools}
\usepackage{nicefrac}
\usepackage{xifthen}
\usepackage{zref-clever}
\zcsetup{cap}

\usepackage[labelformat=simple]{subcaption}

\DeclareCaptionLabelFormat{subfiglabel}{\textsf{\textbf{#2}}}
\newcommand{\mathsc}[1]{\textup{\textsc{#1}}}
\newcommand*{\srFlag}[1]{\emph{#1}}

\newcommand*{\Z}{\mathbb{Z}}

\newcommand*{\bigO}{\mathcal{O}}
\newcommand*{\smallO}{o}

\newcommand*{\graph}{G}
\newcommand*{\vertices}{V}
\newcommand*{\edges}{E}
\newcommand*{\dirN}{\mathsc{n}}
\newcommand*{\dirNE}{\mathsc{ne}}
\newcommand*{\dirSE}{\mathsc{se}}
\newcommand*{\dirS}{\mathsc{s}}
\newcommand*{\dirSW}{\mathsc{sw}}
\newcommand*{\dirNW}{\mathsc{nw}}
\newcommand*{\vecN}{\vv{\dirN}}
\newcommand*{\vecNE}{\vv{\dirNE}}
\newcommand*{\vecSE}{\vv{\dirSE}}
\newcommand*{\vecS}{\vv{\dirS}}
\newcommand*{\vecSW}{\vv{\dirSW}}
\newcommand*{\vecNW}{\vv{\dirNW}}
\newcommand*{\directions}{\mathcal{D}}

\newcommand*{\ShapeReconfigurationProblem}{\phase{srProblem}} 

\newcommand*{\robot}{r}
\newcommand*{\robotnode}[1][]{p\ifthenelse{\isempty{#1}}{}{^{#1}}}
\newcommand*{\config}[1][]{C\ifthenelse{\isempty{#1}}{}{^{#1}}}
\newcommand*{\numTiles}{n}
\newcommand*{\numSupply}{m}
\newcommand*{\columnPointer}{d_{\mathrm{col}}}
\newcommand*{\outsidePointer}{d_{\mathrm{out}}}

\newcommand*{\tiles}[1][]{T\ifthenelse{\isempty{#1}}{}{^{#1}}}
\newcommand*{\srInput}{\mathcal{I}}
\newcommand*{\srTarget}{\mathcal{T}}
\newcommand*{\srTargetTiles}[1][]{\tiles[#1] \cap \srTarget}
\newcommand*{\srSupply}[1][]{\tiles[#1] \setminus \srTarget}
\newcommand*{\srDemand}[1][]{\srTarget \setminus \tiles[#1]}
\newcommand*{\srBoundary}[1][]{B\ifthenelse{\isempty{#1}}{}{_{#1}}}

\newcommand*{\nbrhood}{N}

\newcommand*{\blockAlg}{\emph{block formation}}
\newcommand*{\lineAlg}{\emph{line formation}}
\newcommand*{\mazeAlg}{\emph{maze exploration}}

\makeatletter
\newcounter{phasecnt} 
\newcommand{\newPhase}[2]{%
  \refstepcounter{phasecnt}%
  \phantomsection%
  \begingroup%
    \edef\@currentlabelname{#2}%
    \label{phase:#1}%
  \endgroup%
  \hypertarget{phase:#1}{\textsc{#2}}%
}
\DeclareRobustCommand{\phase}[1]{%
  \hyperlink{phase:#1}{\textsc{\nameref{phase:#1}}}%
}
\makeatother

\usepackage{tikz}
\usetikzlibrary{arrows.meta,calc,shapes.geometric,decorations}

\usepackage{tikzscale}

\definecolor{gridColor}{rgb}{0.8,0.8,0.81}
\definecolor{tileColor}{rgb}{0.85,0.85,0.86} 
\definecolor{pebbleColor}{rgb}{0.51,0.5,0.52} 
\definecolor{pathColor}{rgb}{0.6,0.6,0.61} 

\definecolor{supplyColor}{rgb}{1,0.9,0.7}
\definecolor{targetColor}{rgb}{0.3,0.45,0.95}
\definecolor{targetLineColor}{rgb}{0.8,0.9,1}

\definecolor{uniquePointColor}{rgb}{0.3,0.8,0.25}

\tikzset{
	gridline/.style = {
		line width=0.125mm,
		gridColor,
	},
	tile/.style = {
		draw,
		transform shape,
		shape=regular polygon,
		regular polygon sides=6,
		minimum size=sqrt(3)/3*4cm,
		line width=0.2mm,
		fill=tileColor,
	},
	supplytile/.style = {
		tile,
		fill=supplyColor,
	},
	targettile/.style = {
		tile,
		fill=targetColor,
	},
	pebble/.style = {
	  transform shape,
      shape=regular polygon,
      regular polygon sides=6,
      minimum size=1.7cm,
      line width=0.5mm,
      draw=pebbleColor,
    },
	robot/.style = {
		draw,
		transform shape,
		circle,
		minimum size=0.8cm,
		fill=white,
		line width=0.2mm,
	},
	carrytile/.style = {
		tile,
		minimum size=1.3cm,
	},
	targetline/.style = {
        draw=black,
        line width=1.4mm,
		rounded corners=0.5mm,
    },
    innertargetline/.style = {
        draw=targetLineColor,
        line width=1mm,
		rounded corners=0.5mm,
	},
	nexttile/.style = {
		draw,
		transform shape,
		aligned dash={dash on=5pt, dash off=2pt},
		shape=regular polygon,
		regular polygon sides=6,
		minimum size=1.65cm,
		line width=0.3mm,
	},
	double arrow/.style args={#1 colored by #2 and #3}{
        -stealth,line width=#1,#2, 
        postaction={draw,-stealth,#3,line width=(#1-0.4mm),
                    shorten <=0.4mm, shorten >=0.4mm}, 
    },
	boundary/.style = {
		draw=brown!60!black,
		line width=1.2mm,
	},
	phasenode/.style = {
		draw=none,
		inner sep=2pt,
	},
}

\makeatletter
\def\pgf@dec@dashon{5pt}
\def\pgf@dec@dashoff{5pt}
\pgfkeys{/pgf/decoration/.cd,
  dash on/.store in=\pgf@dec@dashon,
  dash off/.store in=\pgf@dec@dashoff
}
\pgfdeclaredecoration{aligned dash}{start}{
\state{start}[width=0pt, next state=pre-corner,persistent precomputation={
\pgfextract@process\pgffirstpoint{\pgfpointdecoratedinputsegmentfirst}%
\pgfextract@process\pgfsecondpoint{\pgfpointdecoratedinputsegmentlast}%
\pgfmathsetlengthmacro\pgf@dec@dashon{\pgf@dec@dashon}%
\pgfmathsetlengthmacro\pgf@dec@dashoff{\pgf@dec@dashoff}%
\pgfmathsetlengthmacro\pgf@dec@halfdash{\pgf@dec@dashon/2}%
}]{}
\state{pre-corner}[width=\pgfdecoratedinputsegmentlength, next state=post-corner, persistent precomputation={
  \pgfmathparse{int(ceil((\pgfdecoratedinputsegmentlength-\pgf@dec@dashon-\pgf@dec@dashoff)/(\pgf@dec@dashon+\pgf@dec@dashoff)))}%
  \let\pgf@n=\pgfmathresult
  \pgfmathsetlengthmacro\pgf@b%
    {\pgfdecoratedinputsegmentlength/(\pgf@n+1)-\pgf@dec@dashon}%
  \ifdim\pgf@b<\pgf@dec@dashoff\relax%
    \pgfmathparse{int(\pgf@n-1)}\let\pgf@n=\pgfmathresult%
    \pgfmathsetlengthmacro\pgf@b%
      {\pgfdecoratedinputsegmentlength/(\pgf@n+1)-\pgf@dec@dashon}%
  \fi%
  \pgfmathsetlengthmacro\pgf@b{\pgf@b+\pgf@dec@dashon}%
}]{%
  \pgfmathloop
  \ifnum\pgfmathcounter>\pgf@n%
  \else%
    \pgfpathmoveto{\pgfpoint{\pgf@b*\pgfmathcounter-\pgf@dec@halfdash}{0pt}}%
    \pgfpathlineto{\pgfpoint{\pgf@b*\pgfmathcounter+\pgf@dec@halfdash}{0pt}}%
  \repeatpgfmathloop%
  \pgfpathmoveto%
    {\pgfpoint{\pgfdecoratedinputsegmentlength-\pgf@dec@halfdash}{0pt}}%
  \pgfpathlineto%
    {\pgfpointdecoratedinputsegmentlast}
}
\state{post-corner}[width=0pt, next state=pre-corner]{
   \pgfpathlineto{\pgfpoint{\pgf@dec@halfdash}{0pt}}%
}
\state{final}{
  \pgftransformreset%
  \pgfpathlineto{\pgfpointlineatdistance{\pgf@dec@halfdash}{\pgffirstpoint}{\pgfsecondpoint}}%
}
}
\tikzset{aligned dash/.style={
  decoration={aligned dash, #1}, decorate
}}

\graphicspath{{./figures/}}

\title{Tile Reconfiguration by a Finite Automaton}

\author{Jonas Friemel}{Bochum University of Applied Sciences, Germany}{jonas.friemel@hs-bochum.de}{https://orcid.org/0009-0009-6270-4779}{}
\author{David Liedtke}{Paderborn University, Germany}{liedtke@mail.upb.de}{https://orcid.org/0000-0002-4066-0033}{}
\author{Christian Scheffer}{Bochum University of Applied Sciences, Germany}{christian.scheffer@hs-bochum.de}{https://orcid.org/0000-0002-3471-2706}{}

\authorrunning{J. Friemel, D. Liedtke, and C. Scheffer}



\funding{This work was supported by the Deutsche Forschungsgemeinschaft (German Research Foundation, DFG) -- SCHE~1931/4\nobreakdash-1; SCHE~1592/10\nobreakdash-1.}

\nolinenumbers
\begin{document}

\maketitle

\begin{abstract}
Shape formation is one of the most thoroughly studied problems in programmable matter and swarm robotics.
However, in many models, the class of shapes that can be formed is highly restricted due to the particles' limited memory.
In the hybrid model, an active agent with the computational power of a deterministic finite automaton can form shapes by lifting and placing passive tiles on the triangular lattice.
We study the shape reconfiguration problem where the agent additionally has the ability to distinguish so-called target nodes from non-target nodes and needs to form a target shape from the initial tile configuration.
We present a worst-case optimal $\bigO(\numSupply \numTiles)$ algorithm for simply connected target shapes, where~$\numSupply$ is the initial number of unoccupied target nodes and~$\numTiles$ is the total number of tiles.
Furthermore, we show how an agent can reconfigure a large class of target shapes with holes in $\bigO(\numTiles^4)$ steps.
\end{abstract}

\section{Introduction}

In the field of programmable matter, small (possibly nano-scale) particles are envisioned to solve tasks like self-assembling into desired shapes, making coordinated movements, or coating objects~\cite{toffoli1991programmable}.
The particles may be controlled by external stimuli or act on their own with limited computational capabilities.
In the future, programmable matter could become relevant for targeted medical treatments~\cite{becker2020targeted} or as self-assembling structures in environments that are inaccessible by humans such as in space~\cite{jenett2017design}.
There are multiple computational models of programmable matter that differ in the types of particles, their abilities, and the underlying graph structure.
We focus on the \emph{hybrid model} of programmable matter where passive hexagonal tiles and an active agent with the computational power of a deterministic finite automaton populate the triangular lattice~\cite{gmyr2018recognition,gmyr2020forming}.

A central research problem in programmable matter is shape formation~\cite{daymude2019computing}.
In the hybrid model, the active agent has to rearrange passive tiles from an arbitrary initial configuration into a shape~\cite{gmyr2020forming} such as a line or a triangle.
Once the shape is formed, it is typically assumed that the work is finished and the shape will remain intact.
Thus, existing shape formation algorithms are not designed to repair small shape defects such as individually misplaced tiles and may deconstruct the entire structure to rebuild the desired shape.

\subsection{Our Contributions}

We provide first results for shape reconfiguration in the hybrid model with the assumption that the agent is able to recognize nodes that belong to the target shape.
We present a worst-case optimal algorithm for reconfiguring simply connected target shapes within $\bigO(\numSupply \numTiles)$ steps, where~$\numTiles$ is the configuration size and~$\numSupply$ is the number of tiles initially located outside of the target shape (\zcref{sec:target_no_holes}).
We also describe how an agent equipped with a single counter can reconfigure arbitrary shapes with holes in $\bigO(\numSupply \numTiles^2)$ steps and we present an~$\bigO(\numTiles^4)$ algorithm for scaled shapes, i.e., target shapes that do not contain any nodes with multiple disconnected non-target neighbors, that does not require the use of counters or other auxiliary means (\zcref{sec:target_with_holes}).
This is particularly interesting as it is impossible for finite automata to visit all nodes if they are unable to modify their environment~\cite{budach1978automata}.

\subsection{Related Work}

Shape formation has been studied in a multitude of models for programmable matter.
The particles in each model generally fall into one of the following two categories:
\emph{active} agents that can perform (typically limited) computations and move on the underlying graph by themselves, and \emph{passive} entities that do not move or act without external influence.

Models with distributed active agents where shape formation has been studied include:
the geometric \emph{amoebot} model~\cite{derakhshandeh2016universal,derakhshandeh2015leaderelection,diluna2020shapeformation,kostitsyna2022faulttolerant}, where particles with restricted vision occupy either one node (contracted) or two adjacent nodes (expanded) and have the computational capabilities of finite automata~\cite{daymude2019computing,daymude2023canonical,derakhshandeh2014amoebot};
the \emph{silbot} model~\cite{navarra2025line}, where the particles additionally lack persistent memory and means of computation~\cite{dangelo2020leadersilent};
the \emph{nubot} model~\cite{chen2014fastassembly,woods2013selfassembly}, where particles (\emph{monomers}) may form bonds with their neighbors and can move and interact according to predefined rules;
or the \emph{sliding square} model~\cite{dumitrescu2004motion,hurtado2015distributed,kostitsyna2025rhombus} where square particles can move via slides and convex transitions.

A well-known model of passive matter is the \emph{abstract tile assembly model (aTAM)}~\cite{doty2012selfassembly,winfree1998selfassembly}, where tiles with different types of glue with varying strengths occupy a square grid graph.
Starting from a seed configuration, the shape self-assembles with new tiles attaching themselves to compatible existing tiles.
For a survey of the aTAM and variants such as the \emph{kinetic tile assembly model~(kTAM)} and the \emph{two-handed assembly model~(2HAM)}, we refer to Patitz~\cite{patitz2014tileassembly}.
A related passive model is \emph{tilt assembly}, where tiles with uniform glue are dropped onto the seed configuration from the outside with axis-aligned moves~\cite{balanza2020hierarchical,becker2020tilt,keller2022particle,manzoor2017parallel}.

In this paper, we are working with the \emph{hybrid model} introduced by Gmyr et al.~\cite{gmyr2020forming}.
Here, active finite automata move on the triangular lattice and can lift and place passive tiles.
The authors present shape formation algorithms for simple shapes such as lines, blocks, and triangles.
On square grid graphs, finite automata can construct bounding boxes and scale polyominoes using markers~\cite{fekete2021cadbots,fekete2022connectedlattice}.
In the three-dimensional variant of the hybrid model, an agent can coat objects~\cite{kostitsyna2024coating} and construct lines~\cite{hinnenthal2020shape3d} as well as simply connected intermediate shapes called ``icicles''~\cite{hinnenthal2025icicle}.

There has been some research into repairing simple shapes such as lines.
In the amoebot model, Di Luna et al.~\cite{diluna2018linerecovery} and Nokhanji and Santoro~\cite{nokhanji2020linereconfiguration} separately study the problem of repairing a line when some of the particles that are part of the shape malfunction.
Nokhanji et al.~\cite{nokhanji2023linehybrid} also explore a similar problem in the hybrid model where tiles in a line become faulty and need to be removed.
Finally, Kostitsyna et al.~\cite{kostitsyna2023fastreconfiguration} propose a linear-time (in expectation) algorithm to reconfigure simply connected amoebot shapes where particles are pulled to their target positions along shortest path trees.

Real-world implementations of material-assembly systems with active robots include the BILL\nobreakdash-E platform~\cite{jenett2017bille,jenett2019material} or the MMIC\nobreakdash-I~\cite{formoso2023mmic}, SOLL\nobreakdash-E~\cite{park2023solle}, and SOF\nobreakdash-E~\cite{moon2025sofe} robots in NASA's ARMADAS project~\cite{costa2019algorithmic,gregg2024ultralight}.
Finding an optimal reconfiguration schedule in such systems is \NP-hard~\cite{becker2025reconfiguration,garcia2024reconfiguration}.

\section{Preliminaries}

In our setting, the triangular lattice $\graph = (\vertices, \edges)$ is occupied by a single active agent~$\robot$ with limited computational capabilities and a finite number of passive hexagonal tiles, see \zcref{fig:config_compass}.
We call a node $v \in V$ \emph{tiled} if it is occupied by a tile and denote the set of tiled nodes with~$\tiles$.
At any time, a node may be occupied by at most one tile and each tile may only occupy a single node.
Similarly, the agent may only occupy a single (tiled or untiled) node at a time.
Tiles are uniform, i.e., $\robot$ cannot distinguish any two tiles from one another, and~$\robot$ can carry up to one tile.
When carrying a tile, $\robot$ can still enter tiled nodes.
A tuple $\config = (\tiles, \robotnode)$, where $\robotnode \in \vertices$ is the node occupied by~$\robot$, is called a \emph{configuration}.
It has \emph{size}~$|\tiles|$.
All initial and target configurations in this paper have size~$\numTiles$.

We call a node set $S \subseteq \vertices$ \emph{connected} if the induced subgraph $\graph[S]$ is connected and we call~$S$ \emph{simply connected} if the set $\vertices \setminus S$ is connected.
A configuration $\config = (\tiles, \robotnode)$ is \emph{connected} if~$\tiles$ is connected or if $\tiles \cup \{\robotnode\}$ is connected and~$\robot$ is carrying a tile.
Similarly, $\config$ is \emph{simply connected} if $\tiles$ is simply connected, i.e., $\vertices \setminus \tiles$ is connected, or if $\tiles \cup \{\robotnode\}$ is simply connected and~$\robot$ is carrying a tile.
We require configurations to be connected to ensure that the tile system does not drift apart in practical implementations, e.g., in liquid domains.

\begin{figure}[tb]
    \centering%
    \includegraphics[width=0.5\textwidth]{config_compass}%
    \caption{An agent on tiles and global compass directions on the triangular lattice.}%
    \label{fig:config_compass}%
\end{figure}

The agent has an internal compass to differentiate between the six edge directions on the graph~$\graph$ (\dirN, \dirNE, \dirSE, \dirS, \dirSW, \dirNW).
For ease of presentation, we assume that this compass aligns with the global directions on the triangular lattice shown in \zcref{fig:config_compass}.
We denote the set of compass directions by $\directions \coloneq \{\dirN, \dirNE, \dirSE, \dirS, \dirSW, \dirNW\}$ and define $\directions$ to be isomorphic to the ring of integers modulo six $\nicefrac{\Z}{6\Z}$ with $\dirN \equiv 0$, $\dirNE \equiv 1$, and so on, up to $\dirNW \equiv 5$.
With a slight abuse of notation, this allows us to perform simple additions on directions, e.g., $\dirNE + 2 = \dirS$.
Intuitively, by adding $k \in \Z$ to a direction $d \in \directions$, we obtain the next direction from~$d$ after~$k$ clockwise turns of 60 degrees around the compass shown in \zcref{fig:config_compass}.

Each node $v \in \vertices$ is uniquely identified by a coordinate pair $(x, y) \in \Z \times \Z$ where $x + y$ is even.
We write $v = (x, y)$.
The $x$-coordinate grows from west to east, and the $y$-coordinate grows from north to south.
By this convention, the six compass directions correspond to the following directional vectors.
\[%
    \vecN = (0,-2),\quad \vecNE = (1,-1),\quad \vecSE = (1,1),\quad \vecS = (0,2),\quad \vecSW = (-1,1),\quad \vecNW = (-1,-1).
\]%

For a node $v \in \vertices$ and a direction $d \in \directions$, the node $v + \vv{d}$ is called a \emph{neighbor of~$v$} (in direction~$d$) and $\nbrhood(v) \coloneq \{v + \vv{d} \; | \; d \in \directions\}$ is called the \emph{neighborhood of~$v$}.

The agent acts in \emph{look-compute-move} cycles.
In the \emph{look} phase, it observes its surroundings. The agent's ``vision'' is limited to neighboring nodes, i.e., it can only see tiles within unit hop-distance to its node in $\graph$.
Next, $\robot$ enters the \emph{compute} phase where it uses the gathered information to determine its next internal state and its action on the graph.
The agent has the computational capabilities of a deterministic finite-state automaton.
Consequently, it has only constant memory and cannot store a map of the tile configuration.
Finally, $\robot$ enters the \emph{move} phase, where it may perform any of the following actions, provided that connectivity is maintained:
(i)~move to an adjacent (tiled or untiled) node (regardless of whether it is carrying a tile), (ii)~lift the tile at its position if it is not carrying a tile, and (iii)~place a tile at its position if it is carrying a tile and the node is untiled.

\subsection{Problem Statement}

Consider two connected sets of nodes~$\srInput, \srTarget \subseteq \vertices$ with $|\srInput| = |\srTarget| = \numTiles$ and an initial position $p^0 \in \srInput$ for the agent.
We refer to~$\srInput$ and~$\srTarget$ as the \emph{initial} and \emph{target shape}, respectively, with the corresponding nodes being referred to as \emph{initial} and \emph{target nodes}.
Note that a shape is defined by its exact coordinates on the lattice, i.e., two shapes that are translationally or rotationally symmetrical are generally~not~identical.

An agent solves the \newPhase{srProblem}{Shape Reconfiguration Problem} by executing an algorithm that results in a sequence of connected configurations $\config[0] = (\tiles[0], \robotnode[0]), \dots,\allowbreak \config[\ell] = (\tiles[\ell], \robotnode[\ell])$ for some $\robotnode[\ell] \in \vertices$ with $\tiles[0] = \srInput$ and $\tiles[\ell] = \srTarget$ such that each configuration~$\config[t]$ results from configuration~$\config[t-1]$ by applying the agent's legal move actions (i)--(iii) to $\robotnode[t-1]$ for $0 < t \leq \ell$.
When the time step is clear from the context or not relevant, we drop the superscripts.
At any time~$t$, a node $v \in \srSupply[t]$ is called a \emph{supply node} and a node $w \in \srDemand[t]$ is called a \emph{demand node}.
We denote the initial number of supply and demand nodes by $\numSupply \coloneq |\srInput \setminus \srTarget| = |\srTarget \setminus \srInput|$.
Finally, tiles on target nodes are called \emph{target tiles} and tiles on supply nodes are called \emph{supply tiles}.
Thus, to solve the \ShapeReconfigurationProblem{}, an agent needs to move all~$\numSupply$ supply tiles to demand nodes; see \zcref{fig:example_instance}.

\begin{figure}[tb]
    \begin{subfigure}{0.475\textwidth}%
        \centering%
        \includegraphics[width=\textwidth]{example_instance_initial}%
        \caption{}%
        \label{fig:example_instance_initial}
    \end{subfigure}%
    \hfill%
    \begin{subfigure}{0.475\textwidth}%
        \centering%
        \includegraphics[width=\textwidth]{example_instance_finished}%
        \caption{}%
        \label{fig:example_instance_finished}%
    \end{subfigure}%
    \caption{%
        An example instance of the \ShapeReconfigurationProblem{}.
        The light blue line encircles the target shape~$\srTarget$.
        The blue tiles are target tiles.
        The yellow tiles are supply tiles and need to be moved to untiled target nodes (demand nodes).
        \subref{fig:example_instance_initial}:~The positions of the tiles in the initial shape~$\srInput$.
        \subref{fig:example_instance_finished}:~The final shape after all supply tiles have been moved to the target shape~$\srTarget$.
        In this example, the target shape~$\srTarget$ is simply connected while the initial shape~$\srInput$~is~not.
    }%
    \label{fig:example_instance}
\end{figure}

At any time~$t$, the agent can determine whether $\robotnode[t] \in \srTarget$ when it is in the \emph{look} phase of a look-compute-move cycle.
In a practical implementation, this node distinguishability could be realized with a simple binary signal from the outside, e.g., from a light source.
We also assume that the agent can see which of its neighboring nodes are target nodes.
This assumption does not make our agent more powerful than an agent $\robot'$ that can only query~$\srTarget$ for its own position as~$\robot'$ could simply visit all six adjacent nodes within a constant number of steps.
Some adjacent nodes may be unreachable without violating the connectivity requirement, but our agent ignores these nodes in the presented algorithms.

We assume that the initial set of tiled target nodes $\srInput \cap \srTarget$ is non-empty and connected.
This allows the agent to carry supply tiles to demand nodes via paths over target tiles that do not need to be lifted.
Such assumptions are common in the literature, e.g.,~\cite{kostitsyna2023fastreconfiguration}.
Consider the problem instance shown in \zcref{fig:disconnected_intersect}.
Here, the supply tile cannot be moved directly to the demand node without breaking connectivity or without first moving a target tile.

\begin{figure}[tb]
    \centering%
    \includegraphics[width=0.35\textwidth]{disconnected_intersect}%
    \caption{The initial set of tiled target nodes $\srInput \cap \srTarget$ is not connected.}%
    \label{fig:disconnected_intersect}%
\end{figure}

\subsection{Boundaries and Boundary Traversal}

Let $S \subseteq \vertices$ be a finite subset of nodes (a \emph{shape}).
Consider the node sets $M_1, \dots, M_k \subseteq \vertices \setminus S$ of all connected components of $\graph[\vertices \setminus S]$.
If~$S$ is simply connected, then~$k = 1$.
We define the \emph{boundary} of~$M_i$ as $\srBoundary^S(M_i) \coloneq \bigcup_{v \in M_i} \nbrhood(v) \cap S$ and refer to $M_i$ as the \emph{outside} of the boundary $\srBoundary^S(M_i)$.
Let~$M^\ast$ be the unique set of infinite size among the~$M_i$.
We refer to~$M^\ast$ as the \emph{outside} of the shape~$S$, to~$\srBoundary^S(M^\ast)$ as the \emph{main boundary} of~$S$ and to~$\srBoundary^S(M_i)$ as an \emph{inner boundary} for any $M_i \neq M^\ast$.
For ease of presentation, we write $\srBoundary(S) \coloneq \srBoundary^S(M^\ast)$.
Let~$w \in S$ and let $M_w^\ast$ be a set of maximum size (not necessarily infinite) among all~$M_i$ that contain a node adjacent to~$w$ in~$\graph$.
Then $\srBoundary[w](S) \coloneq \srBoundary^S(M_w^\ast)$ is called a $w$\nobreakdash-boundary of~$S$.
Note that $\srBoundary[w](S)$ is not uniquely defined in general.
We refer to~$\srBoundary(\tiles)$ as the \emph{tile boundary}, to~$\srBoundary(\srTarget)$ as the \emph{target boundary}, to~$\srBoundary(\srTargetTiles)$ as the \emph{target tile boundary}, and to $\srBoundary[w](\srSupply)$ as the \emph{boundary of a supply component} for a supply node~$w \in \srSupply$.
See \zcref{fig:boundaries} for examples.

\begin{figure}[tb]
    \begin{subfigure}{0.32\textwidth}%
        \centering%
        \includegraphics[width=\textwidth]{boundaries_target}%
        \caption{}%
        \label{fig:boundaries_target}%
    \end{subfigure}%
    \hfill%
    \begin{subfigure}{0.32\textwidth}%
        \centering%
        \includegraphics[width=\textwidth]{boundaries_target_tile}%
        \caption{}%
        \label{fig:boundaries_target_tile}%
    \end{subfigure}%
    \hfill%
    \begin{subfigure}{0.32\textwidth}%
        \centering%
        \includegraphics[width=\textwidth]{boundaries_supply_component}%
        \caption{}%
        \label{fig:boundaries_supply_component}%
    \end{subfigure}%
    \caption{%
        Boundaries \subref{fig:boundaries_target}~$\srBoundary(\srTarget)$, \subref{fig:boundaries_target_tile}~$\srBoundary(\srTargetTiles)$, and \subref{fig:boundaries_supply_component}~$\srBoundary[\robotnode](\srSupply)$.
        The agent~$\robot$ is on node~$\robotnode$.
    }%
    \label{fig:boundaries}%
\end{figure}

The algorithms presented in this paper rely on the agent~$\robot$'s ability to traverse boundaries by the so-called \emph{left-hand rule}~(LHR):
For some boundary~$B \subseteq S$ of a shape~$S$ with outside~$M$, $\robot$ moves along nodes in~$B$ by always keeping the outside~$M$ to its left.
Concretely, $\robot$ encodes an ``outside pointer'' $\outsidePointer \in \directions$ in its internal states that points to an outside node $w \in M$.
To compute the next movement direction, $\robot$ picks the smallest positive integer $k \in \mathbb{N}$ such that $\robotnode + \vv{\outsidePointer + k} \in S$, moves in direction~$\outsidePointer + k$, and updates~$\outsidePointer$ to $\outsidePointer + k - 2$ to ensure that~$\outsidePointer$ continues to point to an outside node $\tilde{w} \in M$, which is not necessarily equal to~$w$.
The \emph{right-hand rule} (RHR) works analogously.
Both rules are visualized in \zcref{fig:lhr_rhr}.

A hole-free shape~$S$ has only one boundary~$\srBoundary(S)$, so~$\robot$ only needs to find two adjacent nodes $v,w$ with $v \in S$ and $w \in \vertices \setminus S$ to traverse~$\srBoundary(S)$ as~$w$ always belongs to its outside.
Shapes with holes have multiple boundaries, so~$\robot$ cannot locally distinguish a node $w \in \vertices \setminus S$ on the outside of $\srBoundary(S)$ from a node $\tilde{w} \in \vertices \setminus S$ on the outside of some inner boundary~$B' \neq \srBoundary(S)$.

For simplicity, we do not explicitly mention the outside pointer $\outsidePointer \in \directions$ in our algorithm descriptions and instead just write that the agent traverses a boundary by the LHR (or the RHR).
We give more details in proofs when necessary.

\begin{figure}[tb]
    \begin{subfigure}{0.23\textwidth}%
        \centering%
        \includegraphics[width=\textwidth]{lhr_outside_a}%
        \caption{}%
        \label{fig:lhr_outside_a}%
    \end{subfigure}%
    \hfill%
    \begin{subfigure}{0.23\textwidth}%
        \centering%
        \includegraphics[width=\textwidth]{lhr_outside_b}%
        \caption{}%
        \label{fig:lhr_outside_b}%
    \end{subfigure}%
    \hfill%
    \begin{subfigure}{0.23\textwidth}%
        \centering%
        \includegraphics[width=\textwidth]{rhr_outside_a}%
        \caption{}%
        \label{fig:rhr_outside_a}%
    \end{subfigure}%
    \hfill%
    \begin{subfigure}{0.23\textwidth}%
        \centering%
        \includegraphics[width=\textwidth]{rhr_outside_b}%
        \caption{}%
        \label{fig:rhr_outside_b}%
    \end{subfigure}%
    \caption{%
        An agent traversing the tile boundary~$\srBoundary(\tiles)$. \subref{fig:lhr_outside_a}--\subref{fig:lhr_outside_b}:~LHR. \subref{fig:rhr_outside_a}--\subref{fig:rhr_outside_b}:~RHR.
    }%
    \label{fig:lhr_rhr}%
\end{figure}

\section{Simply Connected Target Shapes}
\label{sec:target_no_holes}

We first present a worst-case optimal algorithm to solve the \ShapeReconfigurationProblem{} for simply connected target shapes.
Note that the initial shape~$\srInput$ may contain holes.
We begin by showing how the agent~$\robot$ can find the target tile boundary from any initial position.
This will then allow us to assume that~$\robot$ is initially located at the target tile boundary, i.e., $\robotnode[0] \in \srBoundary(\srTargetTiles)$, and is equipped with a pointer $\outsidePointer \in \directions$ to the outside of $\srBoundary(\srTargetTiles)$, i.e., $\robotnode[0] + \vv{\outsidePointer} \notin \srTargetTiles$, in the remainder of the section.

\subsection{Finding the Target Tile Boundary}
\label{sec:target_no_holes:find_boundary}

We first show that the agent~$\robot$ can find the target tile boundary $\srBoundary(\srTargetTiles)$ and its outside~$M$ from any initial position $\robotnode[0] \in \srInput$ if the target shape~$\srTarget$ is simply connected.
The naive approach of moving in a fixed direction does not work as the tile shape may initially contain holes that prevent the agent from reaching the target tile boundary.
We present strategies for exploring shapes with holes in \zcref{sec:target_with_holes}, but they come with a significant runtime penalty.
Instead, for the first algorithm, we make use of the fact that $\srTarget$ is simply connected, which implies that any tiled target node $v \in \srTargetTiles$ with a non-target node neighbor $w \in \vertices \setminus \srTarget$ lies on the target tile boundary, i.e., $v \in \srBoundary(\srTargetTiles)$ and~$w \in M$.

Thus, the agent can execute an existing algorithm where all tiled nodes are eventually visited, and stop as soon as it finds a tiled target node with a non-target neighbor to initialize the outside pointer for the traversal.
We use the \blockAlg{} algorithm by Gmyr et al.~\cite[Theorem~4]{gmyr2020forming} because of its simplicity, but other shape formation algorithms would work as well.
Note that~$\robot$ stops as soon as it moves a tile from a target node to a non-target node (or the other way around), so the number of supply tiles remains $\numSupply \pm 1 = \bigO(\numSupply)$ when~$\robot$ stops the shape \emph{formation} algorithm and begins executing the shape \emph{reconfiguration} algorithm.
Since the agent moves some tiles during the execution of \blockAlg{}, we also need to ensure that the target tile shape~$\srTargetTiles$ remains connected afterward, as the reconfiguration algorithm relies on this assumption.

\begin{lemma}
    \label{lem:find_target_boundary}
    The agent can find the target tile boundary in $\bigO(\numSupply \numTiles)$ time steps on instances of the \textsl{\ShapeReconfigurationProblem{}} with simply connected target shapes, maintaining connectivity of the target tile shape.
\end{lemma}
\begin{proof}
    The agent~$\robot$ is done as soon as it either enters a tiled target node with a non-target node neighbor or a supply node with a tiled target node neighbor.
    This trivially happens within $\bigO(\numTiles D)$ time steps where $D$ is the diameter of the initial shape~$\srInput = \tiles[0]$ since \blockAlg{} terminates within $\bigO(\numTiles D)$ time steps during which~$\robot$ visits every tiled node~\cite{gmyr2020forming}.
    Assume $\robot$ is initially placed on a supply node $\robotnode[0] \in \srSupply[0]$.
    Let~$\tilde{\numSupply}$ be the size of the agent's supply component and let $\tilde{D}$ be the component's diameter.
    Then, $\robot$ finds the target tile boundary within $\bigO(\tilde{\numSupply} \tilde{D}) = \bigO(\numSupply^2) = \bigO(\numSupply \numTiles)$ time steps.
    In this case, $\robot$ does not lift any tile from a target node, so $\srTargetTiles$ remains connected.

    Now assume~$\robot$ is initially placed on a tiled target node~$\robotnode[0] \in \srTargetTiles[0]$.
    Observe that at most~$\numSupply + 1$ tiles are moved before the agent stops:
    Each tile is moved to an unoccupied node, which is either one of the~$\numSupply$ demand nodes or an empty non-target node.
    As soon as a tile is moved to a non-target node, the agent stops.
    While it is possible that a tile is moved to a demand node that was previously occupied by a different tile, that tile must have been moved to a different demand node before.

    Additionally, each of the moved tiles is moved by no more than~$\numTiles$ steps:
    In \blockAlg{}, tiles are moved in a fixed direction (\dirSE) to the next unoccupied node.
    If a tile is moved farther than~$\numTiles$ from its starting position, at least one node on the tile's path must be a non-target node and the agent stops.

    Thus, the number of time steps during which~$\robot$ carries a tile is bounded by $(\numSupply + 1) \cdot \numTiles = \bigO(\numSupply \numTiles)$.
    Finally, Gmyr et al.~\cite{gmyr2020forming} already showed that the number of time steps an agent executing \blockAlg{} spends searching for tiles is bounded by $\bigO(t + \numTiles)$ where~$t$ is the number of steps spent carrying tiles.

    It remains to argue that the agent maintains connectivity of $\srTargetTiles$.
    First note that \blockAlg{} is connectivity-preserving.
    Thus, as $\robotnode^0 \in \srTargetTiles[0]$, the target tile shape $\srTargetTiles$ can only become disconnected if~$\robot$ carries a tile from a target node $v \in \srTarget$ to a non-target node $w \in \vertices \setminus \srTarget$.
    At this point, $\robot$ has found $\srBoundary(\srTargetTiles)$ on the path from $v$ to $w$.
    To restore connectivity, $\robot$ lifts the tile at~$w$, moves in direction $\dirNE$ until $\robotnode \in \srDemand$, i.e., $\robotnode = v$, and places its carried tile.
    Then, $\robot$ moves back in direction $\dirSE$ until $\robotnode + \vecSE \notin \srTarget$, ending at $\srBoundary(\srTargetTiles)$.
\end{proof}

\subsection{Simply Connected Reconfiguration}
\label{sec:target_no_holes:algorithm}

On a high level, the algorithm works as follows:
The agent first traverses the target tile boundary~$\srBoundary(\srTargetTiles)$ by the LHR until it finds a connected component of supply tiles (a \emph{supply component}), see \zcref{fig:traversal}.
This happens because the target shape~$\srTarget$ is simply connected, which ensures that every supply component is adjacent to~$\srBoundary(\srTargetTiles)$.
The agent moves to the component and traverses it until it finds a safely removable supply tile, i.e., a supply tile that can be lifted and carried away without violating connectivity.
A deterministic finite automaton cannot always find safely removable tiles on tile shapes with holes whereas finding tiles that can safely be moved to adjacent nodes is possible~\cite{gmyr2020forming}.
Therefore, we require the agent to reconfigure the supply component while looking for removable tiles.
In particular, it compacts the supply component by moving supply tiles away from the outside of the component's boundary~$\srBoundary[\robotnode](\srSupply)$ whenever possible.
This way, it ``creates'' safely removable tiles which can always be lifted without disconnecting the supply component; see \zcref{fig:supply_compaction}.
A supply node $v \in \srSupply$ is called \emph{free} if the set of supply nodes adjacent to~$v$ is connected.
A tile on a free supply node is called a \emph{free supply tile}.

After the agent lifts such a tile, it returns to the target shape and traverses it until it reaches a demand node where it can place its carried tile.
Traversing the boundary of the target shape is not sufficient as some components of demand nodes may be fully enclosed by tiled target nodes.
To find them, the agent traverses all \emph{columns} of the target shape, which are defined as follows:
For a direction~$d$ and a shape~$S \subseteq \vertices$, a $d$-column is a path of nodes $(v_1, \dots, v_k)$ with $v_i \in S$ such that $v_{i+1} = v_i + \vv{d}$ for every $1 \leq i < k$ and $v_1 + \vv{d + 3}, v_k + \vv{d} \notin S$.
The node~$v_1$ is called the \emph{start} and the node~$v_k$ is called the \emph{end} of the~$d$\nobreakdash-column.
Since~$\srTarget$ is simply connected, all column start nodes lie on the target boundary~$\srBoundary(\srTarget)$.
Thus, the agent can simply traverse~$\srBoundary(\srTarget)$ by the LHR and traverse a column whenever it enters the column's start node until it eventually finds a demand node where it can place its carried tile, see \zcref{fig:traversal}.
After a demand node is found, the agent repeats the above steps to move the next supply tile to a demand node.

\begin{figure}[tb]
    \begin{subfigure}{0.47\textwidth}%
        \centering%
        \includegraphics[width=\textwidth]{find_supply}%
        \caption{}%
        \label{fig:traversal_a}%
    \end{subfigure}%
    \hfill%
    \begin{subfigure}{0.47\textwidth}%
        \centering%
        \includegraphics[width=\textwidth]{traverse_columns}%
        \caption{}%
        \label{fig:traversal_b}%
    \end{subfigure}%
    \caption{%
        \subref{fig:traversal_a}:~The agent traverses the target tile boundary~$\srBoundary(\srTargetTiles)$ by the LHR until it finds a supply tile.
        \subref{fig:traversal_b}:~After lifting a supply tile, the agent traverses the target boundary $\srBoundary(\srTarget)$ and the target shape's columns in phases \phase{traverseBoundaryNoHoles} and \phase{traverseColumnNoHoles} until it reaches a demand node.
        Here, $\columnPointer = \dirN$.
    }%
    \label{fig:traversal}%
\end{figure}

\begin{figure}[tb]
    \centering%
    \includegraphics{phase_transitions_no_holes}%
    \caption{Phase transitions of the algorithm for simply connected target shapes.}%
    \label{fig:phase-transitions-no-holes}%
\end{figure}

\begin{figure}[tb]
    \begin{subfigure}{0.31\textwidth}%
        \centering%
        \includegraphics[width=\textwidth]{supply_compaction_a}%
        \caption{}%
        \label{fig:supply_compaction_a}%
    \end{subfigure}%
    \hfill%
    \begin{subfigure}{0.31\textwidth}%
        \centering%
        \includegraphics[width=\textwidth]{supply_compaction_b}%
        \caption{}%
        \label{fig:supply_compaction_b}%
    \end{subfigure}%
    \hfill%
    \begin{subfigure}{0.31\textwidth}%
        \centering%
        \includegraphics[width=\textwidth]{supply_compaction_c}%
        \caption{}%
        \label{fig:supply_compaction_c}%
    \end{subfigure}%
    \caption{%
        Phase \phase{compactSupplyNoHoles}.
        The agent~$\robot$ traverses~$\srBoundary[\robotnode](\srSupply)$ (gray dashed line) until it enters a supply tile in~\subref{fig:supply_compaction_a} that can be moved inward (node~$v_i$ in the proof of \zcref{lem:compact_supply}).
        Direction~$d = \dirN$ is the next LHR movement direction, but the highlighted node in direction $d + 1 = \dirNE$ is untiled, so~$\robot$ moves its current tile there and then continues its traversal by moving in direction $d - 1 = \dirNW$.
        In~\subref{fig:supply_compaction_b}, the agent continues the same process for the following two tiles.
        Finally, $\robot$ ends on a free supply node in~\subref{fig:supply_compaction_c} (node~$v_j$ in the proof of \zcref{lem:compact_supply}).
    }%
    \label{fig:supply_compaction}
\end{figure}

We now give a more detailed description of the algorithm, divided into four phases (in addition to the initial target boundary location phase, see \zcref{sec:target_no_holes:find_boundary}).
The agent starts in phase \phase{findSupplyNoHoles} with $\robotnode[0] \in \srBoundary(\srTargetTiles)$ and is equipped with a pointer to the boundary's outside.
See \zcref{fig:phase-transitions-no-holes} for a diagram of the algorithm's phase transitions.

\begin{itemize}
    \item \newPhase{findSupplyNoHoles}{Find\-Supply}:
        The agent traverses the target tile boundary $\srBoundary(\srTargetTiles)$ by the LHR until it is adjacent to a supply node $v \in \srSupply$.
        It moves to~$v$ and enters phase \phase{compactSupplyNoHoles}.
    \item \newPhase{compactSupplyNoHoles}{Com\-pact\-Supply}:
        The agent~$\robot$ moves one step along the boundary $\srBoundary[\robotnode](\srSupply)$ of the supply component by the LHR and computes the direction~$d$ of the next node of the boundary traversal.
        If the node in direction $d + 1$ is neither tiled nor a target node, the node in direction $d + 2$ holds a supply tile, and none of the nodes in directions $d - 1, \dots, d - 3$ hold a supply tile, $\robot$ lifts the tile at its current position, moves in direction~$d + 1$, places its carried tile, and moves in direction~$d - 1$.
        Otherwise, $\robot$ just moves in direction~$d$.
        As soon as~$\robot$ enters a free supply node, it lifts the free supply tile, moves to an adjacent supply node (if one exists), and traverses the boundary of the supply component~$\srBoundary[\robotnode](\srSupply)$ by the LHR until it is adjacent to a tiled target node $v \in \srTargetTiles$.
        Then, it stores the direction to~$v$ as~$\columnPointer$, moves to~$v$ and enters phase \phase{traverseColumnNoHoles}.
    \item \newPhase{traverseColumnNoHoles}{Tra\-verse\-Column}:
        The agent~$\robot$ moves in direction~$\columnPointer$ until it either enters an untiled node, in which case~$\robot$ places the tile it is holding, or the node in direction~$\columnPointer$ is not a target node.
        Then, $\robot$ moves in direction~$\columnPointer + 3$, i.e., the opposite direction of~$\columnPointer$, until there is no tiled target node in that direction. If~$\robot$ carries a tile, it enters phase~\phase{traverseBoundaryNoHoles}. Otherwise, it enters phase \phase{findSupplyNoHoles}.
    \item \newPhase{traverseBoundaryNoHoles}{Tra\-verse\-Bound\-ary}:
        The agent~$\robot$ traverses the target boundary~$\srBoundary(\srTarget)$ (at least one step) until it either enters an untiled node, in which case~$\robot$ places the carried tile and enters phase \phase{findSupplyNoHoles}, or the node in direction~$\columnPointer + 3$ is not a target node, in which case~$\robot$ enters phase \phase{traverseColumnNoHoles}.
\end{itemize}

\subsection{Correctness and Runtime of Simply Connected Reconfiguration}
\label{sec:target_no_holes:analysis}

We consider an agent~$\robot$ executing the algorithm given above on an arbitrary instance of the \ShapeReconfigurationProblem{} with a connected initial shape~$\srInput$ and a simply connected target shape~$\srTarget$, assuming that the number of supply tiles is $\numSupply > 0$;
otherwise, the problem is already solved.
We observe~$\robot$'s behavior in each phase in a series of lemmas before combining the results to show the correctness of the algorithm.

\begin{lemma}
    \label{lem:find_supply}
    If the agent is in phase \textsl{\phase{findSupplyNoHoles}} and the number of supply tiles is $\numSupply > 0$, it switches to phase \textsl{\phase{compactSupplyNoHoles}} within~$\bigO(\numTiles)$ time steps.
\end{lemma}
\begin{proof}
    The agent~$\robot$ is initially in phase \phase{findSupplyNoHoles} and it only re-enters the phase from phases \phase{traverseColumnNoHoles} or \phase{traverseBoundaryNoHoles}.
    In all three cases, $\robot$ is at the target tile boundary~$\srBoundary(\srTargetTiles)$.
    The first case is by assumption.
    For the other two cases, see the proof of \zcref{lem:traversal}.
    As all supply components are adjacent to~$\srBoundary(\srTargetTiles)$, $\robot$ finds a supply node during its LHR traversal.
\end{proof}

\begin{lemma}
    \label{lem:compact_supply}
    If the agent is in phase \textsl{\phase{compactSupplyNoHoles}}, it switches to phase \textsl{\phase{traverseColumnNoHoles}} carrying a tile within $\bigO(\numTiles)$ time steps.
\end{lemma}
\begin{proof}
    The agent~$\robot$ enters phase \phase{compactSupplyNoHoles} when it moves from a tiled target node to a supply node in phase \phase{findSupplyNoHoles}, i.e., ${\robotnode[t - 1] \in \srTargetTiles}$ and ${\robotnode[t] \in \srSupply[t]}$.
    Now, $\robot$ sets an internal variable $\outsidePointer \in \directions$ pointing in the direction of $\robotnode[t - 1]$.
    Note that $\outsidePointer$ points to the outside of $\srBoundary[{\robotnode[t]}](\srSupply)$, allowing~$\robot$ to traverse this boundary by the LHR.
    During the traversal, $\outsidePointer$ is continuously updated.
    We first show that~$\robot$ encounters a free supply node during phase \phase{compactSupplyNoHoles}.

    Let $P = (v_1, \dots, v_k)$ with $v_1 = \robotnode[t]$ be the path of an LHR traversal around the supply component boundary $\srBoundary[{\robotnode[t]}](\srSupply)$ where the traversal repeats after~$v_k$.
    Let $\alpha_i \in (-180, 180]$ be the degree by which~$\robot$ has to turn to the right to move from $v_i$ to $v_{i+1}$ (or $v_k$ to $v_1$).
    A negative~$\alpha_i$ corresponds to a left turn with degree~$|\alpha_i|$.
    All possible turns are~shown~in~\zcref{fig:turn_degrees}.

    \begin{figure}[tb]
        \centering%
        \includegraphics[width=0.4\textwidth]{turn_degrees}%
        \caption{%
            The possible turn degrees~$\alpha_i$ and the corresponding successor nodes~$v_{i+1}$ in the path~$P$ in the proof of~\zcref{lem:compact_supply}.
            The darker gray tiles are the agent's predecessor nodes on~$P$ up to~$v_i$ and the lighter gray tiles are potential successors~$v_{i+1}$.
            The two untiled nodes to the southwest of~$v_i$ are on the outside of the boundary and cannot hold tiles in an LHR traversal.
            We can see that $\alpha_i \in \{-60, 0, 60, 120, 180\}$.
            The agent chooses the smallest~$\alpha_i$ such that the corresponding node is tiled.
            Thus, once~$\alpha_i$ and~$v_i$ are chosen, any node in turn direction $\alpha_i' < \alpha_i$ is untiled.
        }%
        \label{fig:turn_degrees}%
    \end{figure}

    Note that if $\alpha_i = 120$ or $\alpha_i = 180$ for any~$i$, the corresponding node~$v_i$ is a free supply node.
    This can be seen in \zcref{fig:turn_degrees}:
    All nodes before~$\alpha_i$ in clockwise order (with $\alpha_i' < \alpha_i)$ are untiled and connected.
    A left turn with $\alpha_i = -120$ is impossible since~$v_{i+1}$ is adjacent to~$v_{i-1}$ and would be the LHR successor of~$v_{i-1}$ instead of~$v_i$.

    Thus, it remains to show that~$\robot$ encounters a free supply node if $|\alpha_i| \leq 60$ for all $i \in \{1,\dots,k\}$.
    Since~$P$ is a clockwise circular path, the sum of the~$\alpha_i$ is $360$.
    At least six right turns are necessary to complete a cumulative $360$~degree turn and every left turn only increases the number of required right turns, so there must be at least five ``two-streaks'' of consecutive right turns, i.e., there must be five pairs $(i, j)$ with $\alpha_i = \alpha_j = 60$ and $\alpha_{i'} = 0$ for $i < i' < j$.
    Let $(i, j)$ be such a pair with the agent~$\robot$ being positioned on~$v_i$ and let~$d$ be the next movement direction of~$\robot$ as in the description of phase \phase{compactSupplyNoHoles}, i.e., $v_i + \vv{d} = v_{i + 1}$.
    Since $\alpha_i = 60$, $\robot$ entered node~$v_i$ from direction $d + 2$.
    By the LHR, $v_i$ can only have supply node neighbors in directions~$d$, $d + 1$, and $d + 2$.
    If there is a tile in direction $d + 1$, $v_i$ is free.
    Thus, assume $v_i + \vv{d + 1}$ is untiled.
    If $v_i + \vv{d + 1} \in \srTarget$, then the node $v_i + \vv{d + 1}$ is not part of a hole in the supply component, so $\robot$ does not need to perform a compaction and can continue its traversal.
    Otherwise, $\robot$, lifts the tile at~$v_i$, moves in direction $d + 1$ and places the tile.
    Then, $\robot$ moves in direction~$d - 1$ and enters~$v_{i + 1}$ as visualized in \zcref{fig:supply_compaction_b}.
    It can easily be verified that $\vv{d + 1} + \vv{d - 1} = \vv{d}$ for all $d \in \directions$.

    Now note that if $\alpha_{i + 1} = 0$, $d$ remains unchanged, $v_{i+1} + \vv{d + 1}$ must be untiled by the same arguments as before, and~$\robot$ entered~$v_{i+1}$ from direction $d + 2$ (opposite direction of $d - 1$), so~$\robot$ repeats the same movement pattern as on~$v_i$.
    This remains true for all $v_{i'}$ with~$i < i' < j$.
    In \zcref{fig:supply_compaction_b}, the highlighted nodes show where the tiles from the nodes~$v_{i'}$ are moved and the gray dashed line shows the agent's movement path.
    After~$\robot$ enters node~$v_j$, it has only two directly adjacent supply neighbors.
    Therefore, $\robot$ has reached a free supply node, see \zcref{fig:supply_compaction_c}.

    Next, we show that the agent does not disconnect the configuration by lifting a free supply tile.
    Note that connectivity can only be violated if this action disconnects the supply component from the target tile shape.
    The supply component remains connected by the definition of a free supply node.
    Without loss of generality, $|\srBoundary[{\robotnode[t]}](\srSupply)| > 1$.
    Otherwise, connectivity is not lost when lifting the only tile of the supply component and the lemma follows immediately.
    After entering phase \phase{compactSupplyNoHoles}, the agent immediately leaves the entry node $\robotnode[t] \in \srSupply$.
    As argued above, $\robot$ will find a free supply node before finishing a full LHR traversal of $\srBoundary[{\robotnode[t]}](\srSupply)$ and entering~$\robotnode[t]$ again.
    Thus, the tile at~$\robotnode[t]$ is not lifted and the supply component remains connected to the rest of the tile shape via~$\robotnode[t]$.

    Finally, we show that~$\robot$ switches to phase \phase{traverseColumnNoHoles} within~$\bigO(\numTiles)$ time steps after lifting the free supply tile.
    Note that the supply node $w \in \srSupply$ on which~$\robot$ enters phase \phase{compactSupplyNoHoles} is adjacent to the target tile boundary $\srBoundary(\srTargetTiles)$.
    As argued above, the tile at~$w$ is only lifted if the supply component contains no other nodes, in which case~$\robot$ is already adjacent to a tiled target node.
    Otherwise, after lifting a tile from a different node in phase \phase{compactSupplyNoHoles}, $\robot$ continues to traverse $\srBoundary[\robotnode](\srSupply)$ by the LHR and eventually returns to~$w$ where it moves to an adjacent tiled target node and switches to phase \phase{traverseColumnNoHoles}.
\end{proof}

\begin{lemma}
    \label{lem:traversal}
    If the agent is in phase \textsl{\phase{traverseColumnNoHoles}} or \textsl{\phase{traverseBoundaryNoHoles}}, it switches to phase \textsl{\phase{findSupplyNoHoles}} within~$\bigO(\numTiles)$ time steps.
\end{lemma}
\begin{proof}
    The agent only enters phase \phase{traverseBoundaryNoHoles} from phase \phase{traverseColumnNoHoles}.
    Let~$t$ be a time at which~$\robot$ enters phase \phase{traverseColumnNoHoles}.
    Note that~$\robot$ always carries a tile when entering this phase and that it enters the phase at the start of a $\columnPointer$-column, i.e., $\robotnode[t] \in \srBoundary(\srTarget)$ and $\robotnode[t] + \vv{\columnPointer + 3} \notin \srTarget$.
    In phase \phase{traverseColumnNoHoles}, $\robot$ moves strictly in direction~$\columnPointer$ and all nodes on its path are target nodes.
    The agent stops either if it reaches the end of its column, i.e., the node in direction~$\columnPointer$ is not a target node, or if it enters a demand node.
    In the second case, $\robot$ places its carried tile.
    In both cases, $\robot$ moves back in direction~$\columnPointer + 3$ until it has returned to the column's start node~$\robotnode[t]$.

    If~$\robot$ has placed its tile during the column traversal, it switches to phase \phase{findSupplyNoHoles}.
    Otherwise, it switches to phase \phase{traverseBoundaryNoHoles} where it moves at least one step along~$\srBoundary(\srTarget)$ by the LHR.
    This is possible since~$\robot$ enters this phase from a column start node, i.e., the node $\robotnode + \vv{\columnPointer + 3} \notin \srTarget$ is on the outside of $\srBoundary(\srTarget)$ and~$\robot$ can set~$\outsidePointer$ accordingly.
    As soon as~$\robot$ enters the start node of another $\columnPointer$-column, it traverses that column as explained above.
    There is at least one demand node because~$\robot$ carries a tile and~$|\srTarget| = n$ and every demand node lies on some column whose start node lies on the target boundary~$\srBoundary(\srTarget)$ because~$\srTarget$ is simply connected.
    Therefore, $\robot$ eventually enters a demand node where it can place its tile.
   Then, it returns to the start node of the column, which lies on $\srBoundary(\srTarget)$ and thus also on $\srBoundary(\srTargetTiles)$, and enters phase \phase{findSupplyNoHoles}.

    Since $|\srBoundary(\srTarget)| = \bigO(\numTiles)$, the agent spends~$\bigO(\numTiles)$ time steps in phase \phase{traverseBoundaryNoHoles}.
    Each column is traversed twice in each direction whenever the agent reaches the start node of that column.
    This yields an additional $\bigO(c \numTiles)$ time steps, where~$c$ is the maximum number of times a column start node can be visited by~$\robot$ during a single $\srBoundary(\srTarget)$-traversal.
    It is easy to see that $c = \bigO(1)$, since for any column start node $v \in \srTarget$, the neighborhood~$\nbrhood(v)$ contains at most three disconnected non-target nodes along which~$\robot$ can enter~$v$ by the LHR.
    Thus, $\robot$ enters phase \phase{findSupplyNoHoles} within $\bigO(\numTiles)$ time steps.
\end{proof}

The correctness and runtime of the algorithm follow.

\begin{theorem}
    \label{thm:target_no_holes}
    The agent can solve an instance of the \textsl{\ShapeReconfigurationProblem{}} with a simply connected target shape in $\bigO(\numSupply \numTiles)$ time steps.
\end{theorem}
\begin{proof}
    Let~$\numSupply$ be the initial number of supply tiles.
    If~$\numSupply = 0$, the problem is already solved.
    Assume~$\numSupply > 0$.
    The agent finds the target tile boundary~$\srBoundary(\srTargetTiles)$ and initializes its outside pointer once with a runtime cost of~$\bigO(\numSupply \numTiles)$ (\zcref{lem:find_target_boundary}).
    It then spends $\bigO(\numTiles)$ time steps finding a removable supply tile (phases \phase{findSupplyNoHoles} and \phase{compactSupplyNoHoles}, \zcref{lem:find_supply,lem:compact_supply}), moving it to a suitable demand node (phases \phase{compactSupplyNoHoles}, \phase{traverseColumnNoHoles}, and \phase{traverseBoundaryNoHoles}, \zcref{lem:compact_supply,lem:traversal}), and returning to the target tile boundary (\zcref{lem:compact_supply,lem:traversal}).
    This process is repeated~$\numSupply$ times, i.e., once for each supply tile, bringing the total runtime bound to $\bigO(\numSupply \numTiles)$.
\end{proof}

Existing shape formation algorithms typically have runtimes of $\bigO(\numTiles^2)$ or $\bigO(\numTiles D)$ where~$D$ is the diameter of the initial shape~\cite{gmyr2020forming,hinnenthal2025icicle}.
Thus, due to the agent's ability to distinguish target from non-target nodes, our algorithm is faster than existing algorithms for $\numSupply = \smallO(\numTiles)$, i.e., if the initial and target shape largely overlap.
Furthermore, it is worst-case optimal:
Translating a line of~$\numTiles$ tiles by~$\numSupply$ positions requires~$\Omega(\numSupply \numTiles)$ steps, see \zcref{fig:lower-bound}.

\begin{figure}[tb]
    \centering%
    \includegraphics[width=0.85\textwidth]{lower_bound}%
    \caption{The agent~$\robot$ requires $\Omega(\numSupply \numTiles)$ steps to move all $\numSupply = 5$ supply tiles to demand nodes.}%
    \label{fig:lower-bound}%
\end{figure}

\begin{theorem}
    \label{thm:lower_bound}
    The agent requires $\Omega(\numSupply \numTiles)$ time steps to solve the \textsl{\ShapeReconfigurationProblem{}} in general.
\end{theorem}
\begin{proof}
    Consider a problem instance where all~$\numTiles$ tiles form a straight line.
    The~$\numSupply$ supply tiles are at one end of the line of tiles, the $\numTiles - \numSupply$ target tiles are on the other end, and the~$\numSupply$ demand nodes extend the line beyond the last target tile.
    The agent~$\robot$ is initially placed on the target end of the line.
    Such an instance is depicted in \zcref{fig:lower-bound} with $\numTiles = 8$ and $\numSupply = 5$.

    The only supply tile that can be lifted without violating connectivity is the tile at the other end of the line.
    The agent~$\robot$ needs to move $\numTiles - 1$ steps to get there.
    Once the tile is lifted, $\robot$ needs to move back~$\numTiles$ steps to the nearest demand node to place the tile.
    This process has to be repeated for each of the~$\numSupply$ supply tiles, resulting in $\Omega(\numSupply \numTiles)$ time steps.
\end{proof}

\section{Target Shapes with Holes}
\label{sec:target_with_holes}

If the target shape~$\srTarget$ has holes, our previous algorithm fails:
In such a shape, some columns' start nodes are on an inner boundary $B \neq \srBoundary(\srTarget)$, so they may not be found by an agent executing the previous traversal phases.
This failure is not surprising:
In a related setting, a finite automaton cannot visit all cells in a grid maze with holes when it is equipped with fewer than two pebbles that can be used to mark cells~\cite{budach1978automata,hoffmann1981pebblesearch}.
We conjecture that this impossibility extends to the \ShapeReconfigurationProblem{} for arbitrary target shapes with holes.

Thus, we first present a reconfiguration algorithm for an agent that can mark up to two tiles at a time with pebbles.
Then, we show how an agent without this ability can emulate placing pebbles on sufficiently scaled target shapes and present a reconfiguration algorithm based on this approach.

\subsection{Target Shape Exploration}
\label{sec:target_exploration}

A deterministic finite automaton equipped with a single counter of size~$\Theta(\log n)$ can explore arbitrary orthogonal mazes of size~$\numTiles$ in $\bigO(\numTiles^2)$ steps as shown by Blum and Kozen~\cite[Theorem~3.1]{blum1978powercompass}.
If the agent is instead equipped with two placeable markers whose distance in the maze emulates the counter, the execution time is $\bigO(\numTiles^3)$.
This \mazeAlg{} algorithm can be adapted for the hybrid model where an agent equipped with a counter---or two not necessarily distinguishable pebbles that can be placed on tiles---can visit every tile in a connected configuration~\cite{liedtke2021thesis}.
The counter is used to identify \emph{unique points}, which are the westernmost nodes of their respective boundaries' southernmost nodes.
To check whether a boundary node is a unique point, the agent traverses the boundary multiple times and keeps track of its current ``height,'' i.e., the agent's $y$-distance compared to the unique point candidate.
We omit further details of the unique point detection subroutine and only present the three main phases of the algorithm as we only make changes to these phases in our reconfiguration algorithm.

\begin{itemize}
    \item \newPhase{traverseColumnMaze}{Tra\-verse\-Column}:
        The agent~$\robot$ moves to the \dirN-most tile of its current \dirN-column.
        If this node is a unique point, $\robot$ enters phase \phase{traverseBoundaryMaze}.
        Otherwise, $\robot$ enters phase~\phase{returnSouthMaze}.
    \item \newPhase{returnSouthMaze}{Return\-South}:
        The agent~$\robot$ moves to the \dirS-most tile of its current \dirN-column and enters phase~\phase{traverseBoundaryMaze}.
    \item \newPhase{traverseBoundaryMaze}{Tra\-verse\-Bound\-ary}:
        The agent~$\robot$ traverses its current tile boundary~$\srBoundary[\robotnode](\tiles)$ by the~LHR.
        If~$\robot$ enters a unique point from direction~\dirNW{} during the traversal, it enters phase \phase{returnSouthMaze}.
        If~$\robot$ instead enters a node that does not have a tile neighbor to the south, it enters phase \phase{traverseColumnMaze}.
\end{itemize}

Similarly to the target shape exploration phases of the algorithm presented in \zcref{sec:target_no_holes}, the agent traverses boundaries by the LHR and columns whose start nodes lie on the boundaries~\cite{liedtke2021thesis}.
To ensure that all columns are traversed, the agent must traverse \emph{all} boundaries.
This is where the aforementioned unique points come into play:
The agent begins traversing a boundary when it reaches the boundary's unique point and it stops the traversal once this unique point is reached again.
See \zcref{fig:pebble_traversal}.
The exploration problem has thus been reduced to the problem of determining whether a boundary node~$v$~is~a~unique~point.

\begin{figure}[tb]
    \begin{subfigure}{0.3\textwidth}%
        \centering%
        \includegraphics[width=\textwidth]{pebble_traversal_a}%
        \caption{}%
        \label{fig:pebble_traversal_a}%
    \end{subfigure}%
    \hfill%
    \begin{subfigure}{0.3\textwidth}%
        \centering%
        \includegraphics[width=\textwidth]{pebble_traversal_b}%
        \caption{}%
        \label{fig:pebble_traversal_b}%
    \end{subfigure}%
    \hfill%
    \begin{subfigure}{0.3\textwidth}%
        \centering%
        \includegraphics[width=\textwidth]{pebble_traversal_c}%
        \caption{}%
        \label{fig:pebble_traversal_c}%
    \end{subfigure}%
    \caption{%
        The traversal path of an agent executing \mazeAlg{} (without the unique point detection subroutine) on a tile shape with a hole.
        The unique point of the inner boundary is highlighted in green.
        \subref{fig:pebble_traversal_a}:~The agent traverses the outer boundary and its columns until it enters the unique point of an inner boundary.
        \subref{fig:pebble_traversal_b}:~It traverses the inner boundary and its columns until it returns to the boundary's unique point.
        \subref{fig:pebble_traversal_c}:~It continues its traversal of the outer boundary's columns, eventually visiting all tiles.
    }%
    \label{fig:pebble_traversal}%
\end{figure}

A finite automaton executing \mazeAlg{} requires $\bigO(\numTiles^2)$ steps to visit the entire maze if it has a counter of size~$\Theta(\log n)$.
Alternatively, if the agent is equipped with two pebbles, they can be used to emulate the counter by adjusting the distance between them when traversing a boundary, resulting in $\bigO(\numTiles^3)$ time steps to visit all nodes~\cite{blum1978powercompass}.

\begin{theorem}
    An agent equipped with a counter or two pebbles can visit all tiled nodes in a connected configuration and terminate within $\bigO(\numTiles^2)$ or $\bigO(\numTiles^3)$ time steps, respectively.
\end{theorem}

By making use of \mazeAlg{} to find supply components and demand nodes, and phase \phase{compactSupplyNoHoles} from \zcref{sec:target_no_holes:algorithm} to lift supply tiles, we get the following.

\begin{theorem}
    \label{thm:target_holes_counter}
    An agent equipped with a counter or two pebbles can solve an instance of the \textsl{\ShapeReconfigurationProblem{}} and terminate in~$\bigO(\numSupply \numTiles^2)$ or~$\bigO(\numSupply \numTiles^3)$ time steps, respectively.
\end{theorem}

To the best of our knowledge, no published maze exploration algorithm for finite automata achieves execution times below $\bigO(\numTiles^2)$ (with a counter) or $\bigO(\numTiles^3)$ (with pebbles) on mazes with holes.
Therefore, we do not expect faster shape reconfiguration algorithms in the hybrid model without new traversal strategies.

In the remainder of this paper, we show how an agent with no additional capabilities can solve instances of the \ShapeReconfigurationProblem{} by emulating pebbles.
As a first step, we present how an agent can emulate two-pebble exploration algorithms on problem instances with $\srInput = \srTarget$.
To maintain connectivity, our pebble emulation strategy requires all tiles to be safely removable.
We require that the target shape contains no \emph{bottlenecks}, which are defined as follows:
A node~$v \in \srTarget$ is a \emph{bottleneck} if $\nbrhood(v) \setminus \srTarget$ is not connected.
Upscaling arbitrary shapes by a small constant always results in bottleneck-free shapes.
Such scaling assumptions are common in the literature, e.g.,~\cite{fekete2023connected,luchsinger2019self}.

Let $\mathcal{A}$ be a two-pebble exploration algorithm for an agent in the hybrid model where all pebbles are placed on boundary tiles.
Remember that the agent has the ability to distinguish nodes $v \in \srTargetTiles$ from nodes $w \in \srDemand$ when exploring~$\srTarget$.
The core idea behind our approach is shown in \zcref{fig:emulated_pebbles}:
We emulate pebbles by lifting tiles from target nodes and turning them into demand nodes.
Each demand node then corresponds to a pebble.
To explore all tiled nodes, the agent~$\robot$ simply needs to execute~$\mathcal{A}$, with three adjustments:
If there is a pebble query in~$\mathcal{A}$, $\robot$ checks whether~$\robotnode \in \srDemand$;
if a pebble is placed in~$\mathcal{A}$, $\robot$ lifts a tile;
and if a pebble is picked up in~$\mathcal{A}$, $\robot$ places a tile.

\begin{figure}[tb]
    \begin{subfigure}{0.22\textwidth}%
        \centering%
        \includegraphics[width=\textwidth]{emulated_pebbles_a}%
        \caption{}%
        \label{fig:emulated_pebbles_a}%
    \end{subfigure}%
    \hfill%
    \begin{subfigure}{0.22\textwidth}%
        \centering%
        \includegraphics[width=\textwidth]{emulated_pebbles_b}%
        \caption{}%
        \label{fig:emulated_pebbles_b}%
    \end{subfigure}%
    \hfill%
    \begin{subfigure}{0.22\textwidth}%
        \centering%
        \includegraphics[width=\textwidth]{emulated_pebbles_c}%
        \caption{}%
        \label{fig:emulated_pebbles_c}%
    \end{subfigure}%
    \hfill%
    \begin{subfigure}{0.22\textwidth}%
        \centering%
        \includegraphics[width=\textwidth]{emulated_pebbles_d}%
        \caption{}%
        \label{fig:emulated_pebbles_d}%
    \end{subfigure}%
    \caption{%
        \subref{fig:emulated_pebbles_a}:~A tile shape with two pebbles visualized as gray inner hexagons.
        \subref{fig:emulated_pebbles_b}:~The pebbles are emulated by moving tiles at the pebble positions to the outside of the target shape.
        \subref{fig:emulated_pebbles_c}--\subref{fig:emulated_pebbles_d}:~If two pebbles are adjacent, the position of the second pebble is encoded relative to the first pebble's position as visualized by the black arrow.
    }%
    \label{fig:emulated_pebbles}%
\end{figure}

Since an agent executing~$\mathcal{A}$ only places pebbles on boundary tiles with non-target node neighbors, we can simply move the tiles at the pebble locations to those nodes and move them back when the pebble is lifted.
This approach fails if two pebbles share a single non-target neighbor as only one tile can be moved there.
But as such a scenario only arises if the pebbles are within a constant distance of each other, which can be traversed by the agent within a constant number of moves, the offset of the second pebble relative to the first pebble can be encoded within the agent's constantly many states.
The same can be done if emulating two adjacent pebbles would violate connectivity.

\begin{lemma}
    \label{lem:exploration_no_pebbles}
    Let $\tiles = \srTarget$ and let $\nbrhood(v) \setminus \srTarget$ be connected for each $v \in \srTarget$ (no bottlenecks). Then an agent can visit every target node without the use of pebbles in~$\bigO(\numTiles^3)$ time steps.
\end{lemma}

Note that the strategy still works if~$\robot$ is initially carrying a tile.
In this case,~$\robot$ first places its tile on an adjacent untiled node before lifting the tile at the designated pebble location.

\subsection{Building a Tile Depot}
\label{sec:tile_depot}

The pebble emulation construction presented above is useful for target shapes that are finished, but it does not help us solve the \ShapeReconfigurationProblem{} just yet.
An initial shape~$\srInput$ has~$\numSupply$ demand nodes which may all be indistinguishable from emulated pebbles.
Thus, we form an intermediate shape where the agent can distinguish demand nodes created by emulating pebbles from ``regular'' demand nodes which do not act as markers and need to be tiled.

\begin{figure}[tb]
    \begin{subfigure}{0.27084639498432601880877742946708\textwidth}%
        \centering%
        \includegraphics[width=\textwidth]{depot_formation_a}%
        \caption{}%
        \label{fig:depot_formation_a}%
    \end{subfigure}%
    \hfill%
    \begin{subfigure}{0.27084639498432601880877742946708\textwidth}%
        \centering%
        \includegraphics[width=\textwidth]{depot_formation_b}%
        \caption{}%
        \label{fig:depot_formation_b}%
    \end{subfigure}%
    \hfill%
    \begin{subfigure}{0.18620689655172413793103448275862\textwidth}%
        \centering%
        \includegraphics[width=\textwidth]{depot_formation_c}%
        \caption{}%
        \label{fig:depot_formation_c}%
    \end{subfigure}%
    \hfill%
    \begin{subfigure}{0.17210031347962382445141065830721\textwidth}%
        \centering%
        \includegraphics[width=\textwidth]{depot_formation_d}%
        \caption{}%
        \label{fig:depot_formation_d}%
    \end{subfigure}%
    \caption{%
        Tile depot formation.
        \subref{fig:depot_formation_a}:~In phase \phase{buildLineDepot}, a line is formed by moving \dirN-columns eastward until only a single column remains.
        The black dots show nodes where tiles are lifted and the gray arrows show where those tiles are moved.
        A tile is moved one column at a time (first~\dirSE, then~\dirS{} to an untiled node), so it may need to be lifted multiple times.
        The highlighted nodes in the easternmost column show the untiled nodes that the tiles are moved to.
        \subref{fig:depot_formation_b}:~The resulting line from the previous phases is disconnected from the target shape~$\srTarget$, so it is moved back in phase \phase{moveLineNorthWestDepot} until~$\srTargetTiles \neq \varnothing$.
        \subref{fig:depot_formation_c}:~In phase \phase{moveLineSouthDepot}, the \dirN-most tiles of the line are moved south until the \dirN-most node of the line is south of the target shape~$\srTarget$, see~\subref{fig:depot_formation_d}.
    }%
    \label{fig:depot_formation}
\end{figure}

We begin by removing all tiles from the target shape~$\srTarget$ to create a \emph{depot} of tiles which the agent can repeatedly visit to gradually re\nobreakdash-tile the target shape.
One way to achieve this is by creating a line of tiles in $\bigO(\numTiles^2)$ steps (\lineAlg{})~\cite[Theorem~3]{gmyr2020forming} and then moving that line in some fixed direction until $\tiles \cap \srTarget = \varnothing$, which can easily be checked on a line by simply traversing it. The tile depot formation algorithm works in three phases, see~\zcref{fig:depot_formation}.

\begin{itemize}
    \item \newPhase{buildLineDepot}{Build\-Line}:
        The agent~$\robot$ executes \lineAlg{} to turn the configuration into an \dirN-column.
        In short, the agent moves to a locally westernmost northernmost tile, moves one step~\dirSE{} and then moves~\dirS{} until it can place the tile.
        It repeats this process until a single \dirN-column is formed.
        Then, $\robot$ traverses the column one more time to check whether it still overlaps the target shape.
        If so, it skips to the third phase \phase{moveLineSouthDepot}.
        Otherwise, it enters phase \phase{moveLineNorthWestDepot}.
    \item \newPhase{moveLineNorthWestDepot}{Move\-Line\-North\-West}:
        The agent~$\robot$ moves to the line's~\dirS-most tile.
        Then, $\robot$ repeatedly lifts the tile at~$\robotnode$, moves~\dirNW, places its carried tile, and moves~\dirNE{}, until $\robotnode \notin \tiles$.
        Finally,~$\robot$ moves back in direction~\dirSW{} to the \dirN-most tile of the line.
        If any tile was placed on a target node during the loop, $\robot$ enters phase \phase{moveLineSouthDepot}.
        Otherwise, $\robot$ repeats~this~phase.
    \item \newPhase{moveLineSouthDepot}{Move\-Line\-South}:
        The agent~$\robot$ moves to the line's~\dirN-most tile, lifts it, and moves~\dirS{} until~$\robotnode \notin \tiles$.
        Here, $\robot$ places its carried tile.
        If~$\robot$ entered a target node at any point while moving in direction~\dirS, it repeats this phase.
        Otherwise, $\robot$ terminates.
\end{itemize}

\begin{lemma}
    \label{lem:tile_depot}
    Given an instance of the \textsl{\ShapeReconfigurationProblem{}}, the agent can transform the tile shape into an \dirN-column $(v_1, \dots, v_\numTiles)$ with ${v_\numTiles + \vecN \in \srTarget}$ and $v_i \notin \srTarget$ for~$1 \leq i \leq \numTiles$ and terminate within $\bigO(\numTiles^2)$ time steps.
\end{lemma}
\begin{proof}
    In phase \phase{buildLineDepot}, the agent~$\robot$ forms an \dirN-column of tiles within ${t = \bigO(\numTiles^2)}$ time steps~\cite{gmyr2020forming}.
    Instead of terminating, $\robot$ then switches to either phase \phase{moveLineNorthWestDepot} or phase \phase{moveLineSouthDepot} depending on whether the line still intersects the target shape, which is checked with a single line traversal.

    Assume $\srTargetTiles[t] = \varnothing$, i.e., $\robot$ enters phase \phase{moveLineNorthWestDepot}.
    Let $w_0 \in \srInput \cap \srTarget \neq \varnothing$ be a target node that was tiled in the initial configuration.
    We show that at time~$t$, there exists a path $(w_0, \dots, w_k)$ for some $k = \bigO(\numTiles)$ with $w_{i+1} = w_i + \vecSE$ for $0 \leq i < k$ and where $w_k \in \tiles[t]$ is part of the finished line.
    In other words, a tile on node~$w_k$ lies~$k$ steps in \dirSE-direction from~$w_0$.
    By moving the whole line of tiles in direction~\dirNW, the tile from node~$w_k$ will eventually be placed on node~$w_0$, by which we get $w_0 \in \srTargetTiles[t + t'] \neq \varnothing$ for some~$t'$.
    During \lineAlg{}, a tile from~$w_0$ is moved in direction~\dirSE{} to node $w_1 = w_0 + \vecSE$ and then either placed there or moved further south if another tile already occupies~$w_1$.
    Either way, $w_1 \in \tiles$ after the tile from~$w_0$ is placed.
    Analogously, if~$\robot$ lifts and moves the tile from~$w_1$, then $w_2 = w_1 + \vecSE \in \tiles$ after the tile is placed.
    This argument extends for the path $w_0, \dots, w_k$ where~$w_k \in \tiles[t]$ is part of the final line.
    In the correctness proof of \lineAlg{}, the authors show that every tile is moved by a distance of at most~$2\numTiles$ nodes, so~$k = \bigO(\numTiles)$.
    Moving the entire line in direction~\dirNW{} by a single step takes $\bigO(\numTiles)$ time steps.
    Thus, after $t' = \bigO(\numTiles^2)$ time steps, $\srTargetTiles[t + t'] \neq \varnothing$ and~$\robot$ switches to phase \phase{moveLineSouthDepot}.

    The line now partially consists of target nodes.
    In phase \phase{moveLineSouthDepot}, $\robot$ moves to the \dirN-most tile, lifts that tile, moves~\dirS{} until it is south of the \dirS-most node of the line, and places the carried tile.
    If~$\robot$ encounters any target nodes while moving~\dirS, it repeats the same process.
    Otherwise, it terminates.
    In the worst case, the line's start node~$v_1$ is a target node and the target shape~$\srTarget$ is an \dirN-column with end node~$v_1$, i.e., $v_1 \in \srTarget$ and the entire rest of the target shape $\srTarget \setminus \{v_1\}$ is right to the south of~$v_1$.
    In this case, the line of tiles needs to be moved south by $2n - 1$ nodes.
    For every line movement step, $\robot$ traverses the entire line twice to lift the \dirN-most tile and then place it south of the \dirS-most tile, which takes $\bigO(\numTiles)$ steps.
    Thus, phase \phase{moveLineSouthDepot} also finishes within $\bigO(\numTiles^2)$ time steps.
\end{proof}

\subsection{Reconfiguring Scaled Shapes with Holes}

The high-level idea of the full algorithm is as follows:
The agent first forms a tile depot with the algorithm described in \zcref{sec:tile_depot}, see \zcref{fig:full_algorithm_b}.
Then, it repeatedly lifts the \dirS-most tile (which is safely removable) from the depot, traverses the target shape using the pebble-free exploration algorithm from \zcref{sec:target_exploration}, places the tile at the first demand node it encounters, and returns to the depot.
This process is repeated until all tiles from the depot have been moved to demand nodes.

\begin{figure}[tb]
    \begin{subfigure}{0.31\textwidth}%
        \centering%
        \includegraphics[width=\textwidth]{full_algorithm_a}%
        \caption{}%
        \label{fig:full_algorithm_a}%
    \end{subfigure}%
    \hfill%
    \begin{subfigure}{0.31\textwidth}%
        \centering%
        \includegraphics[width=\textwidth]{full_algorithm_b}%
        \caption{}%
        \label{fig:full_algorithm_b}%
    \end{subfigure}%
    \hfill%
    \begin{subfigure}{0.31\textwidth}%
        \centering%
        \includegraphics[width=\textwidth]{full_algorithm_c}%
        \caption{}%
        \label{fig:full_algorithm_c}%
    \end{subfigure}%
    \caption{%
        \subref{fig:full_algorithm_a}:~An instance with a bottleneck-free target shape.
        \subref{fig:full_algorithm_b}:~The tile depot is built.
        \subref{fig:full_algorithm_c}:~The main target boundary~$\srBoundary(\srTarget)$ is filled, enabling pebble emulation.
    }%
    \label{fig:full_algorithm}%
\end{figure}

It remains to show that~$\robot$ can make use of the pebble-free exploration algorithm even when $\tiles \neq \srTarget$.
Note that $\robot$ only places pebbles on boundary nodes and it only does so during the unique point detection subroutine.
Thus, before the unique point detection subroutine is initiated at some target boundary~$B \subseteq \srTarget$, we need to ensure that~$\robot$ has previously placed tiles at every node $v \in B$ and at every adjacent target node $w \in \nbrhood(B) \cap \srTarget$.
This also ensures that $\nbrhood(v) \cap \tiles$ is connected for each boundary node $v \in B$ and consequently that the requirements for the pebble emulation procedures are fulfilled \emph{locally} along the entire boundary~$B$.
To achieve this, $\robot$ tries to fully traverse a boundary~$B$ by the LHR before initiating the unique point detection subroutine.
If it encounters any demand node $w \in B \cup \nbrhood(B)$ during the traversal, it places its carried tile at~$w$, cancels the boundary traversal, and returns to the tile depot by the RHR.
Otherwise, i.e., if no demand nodes were encountered, the requirements for emulating pebbles along~$B$ are fulfilled and~$\robot$ continues with the exploration algorithm.

For ease of presentation, we use two different phases for filling boundaries with tiles.
Immediately after building the tile depot, $\robot$ traverses and fills the main target boundary~$\srBoundary(\srTarget)$ in phase \phase{fillMainBoundaryHoles}.
It makes use of an internal flag \srFlag{boundaryFlag} which is set as soon as~$\robot$ has fully traversed~$\srBoundary(\srTarget)$ and returned to the tile depot without finding a demand node; see \zcref{fig:full_algorithm_c}.
Other boundaries encountered by~$\robot$ are being tiled in phase \phase{fillBoundaryHoles}.
Here,~$\robot$ emulates placing a pebble at the traversal start and it emulates lifting the pebble after the boundary has been traversed or a demand node was encountered.

After the agent places a tile (in the boundary filling or the target traversal phases), it needs to return to the depot.
It does so by backtracking its previous path using the RHR along previously visited (and thus filled) boundaries.

\begin{figure}[tb]
    \centering%
    \includegraphics{phase_transitions_holes}%
    \caption{Phase transitions of the algorithm for target shapes with no bottlenecks.}%
    \label{fig:phase-transitions-holes}%
\end{figure}

We assume that the agent has previously built a tile depot south of the target shape, see \zcref{sec:tile_depot}.
The full algorithm uses seven phases (in addition to the initial depot formation phase).
It starts in phase \phase{liftDepotTileHoles} on some depot tile.
The phase transitions of the algorithm are illustrated in \zcref{fig:phase-transitions-holes}.

\begin{itemize}
    \item \newPhase{liftDepotTileHoles}{Lift\-Depot\-Tile}:
        The agent~$\robot$ moves in direction~\dirS{} until it reaches an \dirS-most tile, which it lifts.
        Then, $\robot$ moves in direction~\dirN{} until it enters a target node.
        If \srFlag{boundaryFlag} is set, $\robot$ switches to phase \phase{traverseColumnHoles}.
        If not, $\robot$ switches to phase~\phase{fillMainBoundaryHoles}.
    \item \newPhase{fillMainBoundaryHoles}{Fill\-Main\-Bound\-ary}:
        The agent~$\robot$ traverses~$\srBoundary(\srTarget)$ by the LHR.
        If~$\robot$ enters a demand node or is adjacent to one while still carrying a tile, it moves there (if necessary), places its carried tile, moves back, traverses~$\srBoundary(\srTarget)$ by the RHR until it enters a node with a supply neighbor in direction~\dirS, and switches to phase \phase{liftDepotTileHoles}.
        If~$\robot$ instead enters a node with a supply neighbor in direction~\dirS{} again without having placed its carried tile, it sets \srFlag{boundaryFlag} and switches to phase \phase{traverseColumnHoles}.
    \item \newPhase{traverseColumnHoles}{Tra\-verse\-Column}:
        The agent~$\robot$ moves in direction~\dirN{} until it is either adjacent to a demand node~$v \in \srDemand$ or its neighbor in direction~\dirN{} is not a target node.
        In the first case, $\robot$ moves to~$v$, places its carried tile, moves back, and switches to phase \phase{returnSouthHoles}.
        Otherwise, $\robot$ switches to phase \phase{fillBoundaryHoles}.
    \item \newPhase{returnSouthHoles}{Return\-South}:
        The agent~$\robot$ moves in direction~\dirS{} until it has no target node neighbor in direction~\dirS.
        If~$\robot$ carries a tile, it switches to phase \phase{traverseBoundaryHoles}.
        Otherwise, $\robot$ switches to phase \phase{returnToDepotHoles}.
    \item \newPhase{traverseBoundaryHoles}{Tra\-verse\-Bound\-ary}:
        The agent~$\robot$ traverses $\srBoundary[\robotnode](\srTarget)$ by the LHR until it either enters the boundary's unique point, in which case~$\robot$ switches to phase \phase{returnSouthHoles}, or the node in direction~\dirS{} is not a target node, in which case~$\robot$ switches to phase \phase{traverseColumnHoles}.
    \item \newPhase{fillBoundaryHoles}{Fill\-Bound\-ary}:
        The agent~$\robot$ emulates placing a pebble and traverses $\srBoundary[\robotnode](\srTarget)$.
        If~$\robot$ is adjacent to a demand node $v \in \srDemand$ with no emulated pebble, it moves to~$v$, places its carried tile, moves back, and traverses $\srBoundary[\robotnode](\srTarget)$ by the RHR until it enters a node with an emulated pebble.
        It emulates lifting the pebble and switches to phase \phase{returnSouthHoles}.
        If~$\robot$ instead enters a node which holds an emulated pebble without having placed its carried tile, it emulates lifting the pebble and switches to phase \phase{traverseBoundaryHoles} if~$\robotnode$ is the boundary's unique point or phase \phase{returnSouthHoles} otherwise.
    \item \newPhase{returnToDepotHoles}{Return\-To\-Depot}:
        The agent traverses $\srBoundary[\robotnode](\srTarget)$ by the RHR until it has a supply neighbor in direction~\dirS, in which case~$\robot$ switches to phase \phase{liftDepotTileHoles}, or it enters the unique point of the boundary~$\srBoundary[\robotnode](\srTarget)$, in which case~$\robot$ switches to phase \phase{returnSouthHoles}.
\end{itemize}

\subsection{Correctness and Runtime of Bottleneck-Free Reconfiguration}

As in \zcref{sec:target_no_holes:analysis}, we analyze the agent's behavior in individual phases before showing correctness and giving a runtime bound for the complete algorithm.

\begin{lemma}
    \label{lem:fill_boundary}
    If the agent enters phase \textsl{\phase{fillBoundaryHoles}} (\textsl{\phase{fillMainBoundaryHoles}}) on node $\robotnode \in B$ for some target boundary~$B \subseteq \srTarget$, then within~$\bigO(\numTiles)$ time steps, it enters~$\robotnode$ again and one of the following is true:
    \begin{enumerate}[(a)]
        \item \label{lem:fill_boundary_a}%
            $\robot$ no longer carries a tile, or
        \item \label{lem:fill_boundary_b}%
            $\robot$ still carries a tile and all nodes in $B \cup \nbrhood(B) \cap \srTarget$ are tiled.
    \end{enumerate}
\end{lemma}

We only show the lemma for phase \phase{fillBoundaryHoles} since it is essentially identical to \phase{fillMainBoundaryHoles} and the differences between the phases, e.g., checking for a southern supply neighbor instead of using an emulated pebble, can easily be accounted for.

\begin{proof}[Proof of \zcref{lem:fill_boundary}]
    After~$\robot$ enters phase \phase{fillBoundaryHoles} on node $\robotnode[t] \in B$ for some target boundary $B \subseteq \srTarget$ at some time~$t$, it sets its internal outside pointer $\outsidePointer \in \directions$ to the direction of any node $w \in \nbrhood(\robotnode[t]) \setminus \srTarget$, allowing it to traverse~$B$ by the LHR.
    Since~$\srTarget$ does not contain any bottlenecks, $\nbrhood(\robotnode[t]) \setminus \srTarget$ is connected and therefore, all nodes in $\nbrhood(\robotnode[t]) \setminus \srTarget$ belong to the outside of the boundary~$B$.
    Then,~$\robot$ emulates placing a pebble.
    This is possible because \phase{fillBoundaryHoles} is only entered from phase \phase{traverseColumnHoles} where~$\robot$ previously ensured that~$\robotnode[t]$ and all of its target node neighbors are tiled.
    Thus, the requirements for emulating a pebble are locally satisfied.
    Now, $\robot$ traverses~$B$ by the LHR.
    Note that nodes are not visited twice during a single boundary traversal because the target shape contains no bottlenecks so every target node can only be entered from one direction.
    Therefore, when~$\robot$ visits~$\robotnode[t]$ again without having switched its phase, it has previously traversed the entire boundary~$B$.
    Also note that~$\robot$ can distinguish~$\robotnode[t]$ using the pebble query procedure from \zcref{sec:target_exploration}.

    Assume $B \cup \nbrhood(B) \cap \srTarget \not\subseteq \srTargetTiles$, i.e., there is a node $w \in B \cup \nbrhood(B) \cap \srTarget$ that is not tiled.
    Since~$w$ is in~$B$ or adjacent to a node in~$B$, $\robot$ finds it while traversing~$B$.
    The agent then places its carried tile on~$w$, which it carries because the previous phase \phase{traverseColumnHoles} is only entered when~$\robot$ carries a tile.
    Afterward, $\robot$ traverses~$B$ by the RHR, so it returns to~$\robotnode[t]$ on the same path it took to get to~$w$.
    (\ref{lem:fill_boundary_a}) follows.

    Now assume $B \cup \nbrhood(B) \cap \srTarget \subseteq \srTargetTiles$.
    It is clear that~$\robot$ does not find a demand node during its traversal of~$B$, so it it still carries a tile upon reaching~$\robotnode[t]$ again.
    (\ref{lem:fill_boundary_b}) follows.
\end{proof}

The first phase after the tile depot is built is \phase{liftDepotTileHoles}, after which~$\robot$ carries a tile.
Then, $\robot$ enters phase \phase{fillMainBoundaryHoles} on the node~$v$ north of the tile depot, i.e., $v + \vecS \in \srSupply$.
By \zcref{lem:fill_boundary}(\ref{lem:fill_boundary_a}), $\robot$ places the tile along $\srBoundary(\srTarget)$ or an adjacent node, returns to~$v$, and enters phase \phase{liftDepotTileHoles} again.
These steps are repeated until~$\robot$ finds no more demand nodes along $\srBoundary(\srTarget)$ and it returns to~$v$ carrying a tile (\zcref{lem:fill_boundary}(\ref{lem:fill_boundary_b})), where it sets \srFlag{boundaryFlag}.
Thus, we get the following lemma.

\begin{lemma}
    \label{lem:first_boundary_filled}
    After the agent enters phase \textsl{\phase{liftDepotTileHoles}} for the first time, it switches only between phases \textsl{\phase{liftDepotTileHoles}} and \textsl{\phase{fillMainBoundaryHoles}} while \srFlag{boundaryFlag} is unset.
    Within~$\bigO(\numTiles^2)$ time steps, it sets \srFlag{boundaryFlag} and $\srBoundary(\srTarget) \cup \nbrhood(\srBoundary(\srTarget)) \cap \srTarget \subseteq \srTargetTiles$ holds.
\end{lemma}

We know from \zcref{sec:target_exploration} that~$\robot$ can explore the entire target shape as long as the target nodes encountered by~$\robot$ during exploration are tiled to allow for pebble emulation.
With the next two lemmas, we show that this condition is locally fulfilled before~$\robot$ executes the exploration phases that rely on it.

\begin{lemma}
    \label{lem:no_demand_phases}
    If the agent carries a tile and is in phase \textsl{\phase{returnSouthHoles}} or \textsl{\phase{traverseBoundaryHoles}}, then it does not encounter a demand node before switching to another phase.
\end{lemma}
\begin{proof}
    The agent~$\robot$ enters phase \phase{returnSouthHoles} from either phase \phase{fillBoundaryHoles} (preceded by \phase{traverseColumnHoles}) or phase \phase{traverseBoundaryHoles} if it is located on a unique point.
    In the first case, $\robot$ does not encounter a demand node since any demand node on the agent's column would have already been found in phase \phase{traverseColumnHoles} and the agent would have placed its carried tile there.
    For the second case, note that~$\robot$ only traverses an inner boundary~$B \subseteq \srTarget$ in phase \phase{traverseBoundaryHoles} if it previously entered~$B$'s unique point~$u_B \in B$ in phase \phase{traverseColumnHoles}.
    Thus, by the same argument, $\robot$ does not encounter a demand node when traversing the \dirN-column with end node~$u_B$ in direction~\dirS.

    Similarly, $\robot$ enters phase \phase{traverseBoundaryHoles} from either \phase{fillBoundaryHoles} or \phase{returnSouthHoles}.
    In the first case, the inner boundary~$B \subseteq \srTarget$ traversed by~$\robot$ is fully tiled due to \zcref{lem:fill_boundary}.
    In the second case, note that in phase \phase{returnSouthHoles}, $\robot$ only traverses a column previously traversed in phase \phase{traverseColumnHoles} as argued before.
    The start node of that column is part of some boundary~$B$ that is either $\srBoundary(\srTarget)$, which is fully tiled by \zcref{lem:first_boundary_filled} since \srFlag{boundaryFlag} is set, or some inner boundary that was previously traversed already, so it is fully tiled by \zcref{lem:fill_boundary}(\ref{lem:fill_boundary_a}).
\end{proof}

\begin{lemma}
    \label{lem:find_demand_node}
    If \srFlag{boundaryFlag} is set and the agent carries a tile, it switches between phases \textsl{\phase{traverseColumnHoles}}, \textsl{\phase{fillBoundaryHoles}}, \textsl{\phase{returnSouthHoles}}, and \textsl{\phase{traverseBoundaryHoles}} until it finds a demand node that does not belong to an emulated pebble within~$\bigO(\numTiles^3)$ time~steps.
\end{lemma}
\begin{proof}
    In all of the four phases, $\robot$ is on a target node $\robotnode \in \srTarget$.
    Thus, as long as~$\robot$ does not encounter a non-pebble demand node $w \in \srDemand$, it effectively executes the target exploration algorithm described in \zcref{sec:target_exploration}.
    The main difference is that instead of querying for a unique point at the end of \phase{traverseColumnHoles}, $\robot$ first traverses the boundary~$B \subseteq \srTarget$ at the unique point candidate node in phase \phase{fillBoundaryHoles}.
    By \zcref{lem:fill_boundary}, $\robot$ either finds a demand node, in which case we do not need to argue any further, or it returns to the unique point candidate and the boundary~$B$ is fully tiled, i.e., $B \cup \nbrhood(B) \cap \srTarget \subseteq \srTargetTiles$.
    In this case, the requirements for emulating pebbles are fulfilled along~$B$ and~$\robot$ continues with the target exploration algorithm as usual.

    At least one demand node exists as~$\robot$ is carrying a tile which previously occupied a non-target node, so~$\robot$ eventually finds a demand node $w \in \srDemand$ while exploring~$\srTarget$.
    This happens within $\bigO(\numTiles^3)$ time steps by \zcref{lem:exploration_no_pebbles}.
\end{proof}

It remains to show that~$\robot$ returns to the depot after placing a tile at node~$w$.
It does so by applying the traversal strategy in reverse.
We argue that~$\robot$ traverses the same boundaries and columns as it did on its way from the depot to~$w$ in opposite direction.

\begin{lemma}
    \label{lem:back_to_depot}
    If \srFlag{boundaryFlag} is set and the agent places a tile, it enters phase \textsl{\phase{liftDepotTileHoles}} within~$\bigO(\numTiles^3)$ time steps.
\end{lemma}
\begin{proof}
    If~$\robot$ is not already at the tile depot and it has placed a tile at time~$t$, it is (about to be) in either phase \phase{fillBoundaryHoles} or \phase{returnSouthHoles}.
    In phase \phase{fillBoundaryHoles}, $\robot$ traverses its current target boundary $\srBoundary[\robotnode](\srTarget)$ by the RHR (instead of the LHR since it is no longer carrying a tile) until it reaches the node on which it previously entered phase \phase{fillBoundaryHoles}, which is distinguishable because the node holds an emulated pebble.
    Then, $\robot$ enters phase \phase{returnSouthHoles}.

    Whenever~$\robot$ is in phase \phase{returnSouthHoles} before finding the tile depot, we claim that~$\robot$ traverses a column in direction~\dirS{} which it previously traversed in direction~\dirN{} in phase \phase{traverseColumnHoles}, i.e., all nodes traversed in phase \phase{returnSouthHoles} are tiled.
    Initially, this invariant holds:
    The agent~$\robot$ enters phase \phase{returnSouthHoles} for the first time after placing its carried tile from either phase \phase{fillBoundaryHoles} (as above) on the end node of an \dirN-column previously traversed in phase \phase{traverseColumnHoles}, or phase \phase{traverseColumnHoles} itself.

    Once~$\robot$ reaches the start of its current \dirN-column (but not yet the tile depot), it switches to phase \phase{returnToDepotHoles}.
    Here, $\robot$ moves around its current target boundary $\srBoundary[\robotnode](\srTarget)$ by the RHR again until it reaches the boundary's unique point where it previously entered phase \phase{traverseBoundaryHoles} from phase \phase{traverseColumnHoles}.
    Thus, the invariant holds.
    By switching between phases \phase{returnSouthHoles} and \phase{returnToDepotHoles}, $\robot$ backtracks its previous path and eventually returns to the node where it entered \phase{traverseColumnHoles} for the first time after lifting the most recently placed tile from the depot.
    Here, $\robot$ switches to phase \phase{liftDepotTileHoles} and the lemma follows.
\end{proof}

As before, the correctness and runtime of the algorithm follow from the lemmas above.

\begin{theorem}
    \label{thm:target_holes}
    The agent can solve an instance of the \textsl{\ShapeReconfigurationProblem{}} with a bottleneck-free target shape and terminate in $\bigO(\numTiles^4)$ time steps.
\end{theorem}
\begin{proof}
    First, by \zcref{lem:tile_depot}, the agent~$\robot$ forms a tile depot, i.e., a line of tiles whose \dirN-most node has a northern target node neighbor, within~$\bigO(\numTiles^2)$ time steps.
    Then, by \zcref{lem:fill_boundary,lem:first_boundary_filled}, $\robot$ fills~$\srBoundary(\srTarget)$ and the adjacent target nodes with tiles before initiating the exploration phases.
    Now, whenever~$\robot$ leaves the depot with a tile, it finds a demand node~$w \in \srDemand$ in phase \phase{traverseColumnHoles} or phase \phase{fillBoundaryHoles} by \zcref{lem:no_demand_phases,lem:find_demand_node} within $\bigO(\numTiles^3)$ time steps.
    In both cases, $\robot$ places its carried tile on~$w$.
    Afterward, $\robot$ finds its way back to the depot by \zcref{lem:back_to_depot}.
    All steps after the tile depot formation are repeated~$\numTiles$ times, i.e., once for each depot tile, until all demand nodes are tiled and the \ShapeReconfigurationProblem{} is solved.

    Checking for termination is simple:
    Once the agent lifts the last remaining supply tile from the depot, it sets an internal \srFlag{terminationFlag}.
    Then, when~$\robot$ places a carried tile while not executing the unique point detection subroutine, it terminates if \srFlag{terminationFlag}~is~set.
\end{proof}

\section{Discussion and Future Work}
\label{sec:conclusion}

We have shown that a single agent can solve the \ShapeReconfigurationProblem{} for simply connected target shapes in worst-case optimal $\bigO(\numSupply \numTiles)$ steps~(\zcref{thm:target_no_holes,thm:lower_bound}) and for bottleneck-free shapes in~$\bigO(\numTiles^4)$ steps~(\zcref{thm:target_holes}).
Additionally, an agent can solve the problem for arbitrary target shapes without constructing an intermediate tile depot in~$\bigO(\numSupply \numTiles^2)$ steps if equipped with a counter or~$\bigO(\numSupply \numTiles^3)$ steps if equipped with two pebbles (\zcref{thm:target_holes_counter}).
It remains an open question whether shape reconfiguration is possible on arbitrary shapes without the use of pebbles or counters.
We believe this is not the case, just like visiting every cell in a grid maze is impossible for a deterministic finite automaton as shown by Budach~\cite{budach1978automata}.
However, proving this conjecture may not be straightforward as the agent in our setting is more powerful than the automaton in Budach's proof due to its ability to reconfigure its environment.

An agent executing the first algorithm for simply connected target shapes does not terminate once the problem is solved.
We can easily adjust the algorithm to enable termination in~$\bigO(\numTiles^2)$ steps by making use of the tile depot formation algorithm presented in \zcref{sec:tile_depot}, but it would be interesting to see whether faster termination is possible.

Finally, a natural follow-up is to examine reconfiguration with multiple cooperating agents.
As a first step omitted from this paper, we developed and implemented an algorithm for hole-free initial and target shapes and agents with a shared sense of direction.
However, for more general settings, the ideas presented here cannot easily be adapted since agents lifting tiles at different positions can unknowingly disconnect the configuration.

\bibliography{bibliography}

\end{document}